\documentclass[twocolumn,journal]{IEEEtran}
\ifCLASSINFOpdf
\else
\fi
\usepackage{amsmath}
\makeatletter
\def\maketag@@@#1{\hbox{\m@th\normalfont\normalsize#1}}
\makeatother
\usepackage{paralist}
\usepackage[usenames,dvipsnames]{xcolor}
\usepackage{amsthm,amsmath,amsfonts,ed,empheq}
\usepackage{xspace}
\usepackage{url}
\usepackage{hyperref}
\usepackage{multirow}
\usepackage{footnote}
\usepackage{amssymb,amsmath,mathtools,fixmath}
\usepackage{algorithmicx,algorithm}
\usepackage{algpseudocode}
\usepackage{cases}
\usepackage{cite}
\usepackage{graphicx,amsfonts,psfrag,url,booktabs}
\usepackage{comment}
\usepackage{cite}
\usepackage{xcolor}
\usepackage[caption=false,font=footnotesize]{subfig}
\usepackage{makecell}
\usepackage{psfrag}
\usepackage{tabularx}
\usepackage[table]{colortbl}
\usepackage{booktabs}
\usepackage{balance}
\usepackage{graphicx}
\usepackage{enumerate}
\ifCLASSOPTIONcompsoc
\usepackage[caption=false,font=normalsize,labelfont=sf,textfont=sf]{subfig}
\else
\usepackage[caption=false,font=footnotesize]{subfig}
\fi

\newtheorem{property}{Property}
\newtheorem{proposition}{Proposition}

\begin{document}

\title{Filter Design for Autoregressive Moving Average Graph Filters}

\author{Jiani~Liu,~\IEEEmembership{Student Member,~IEEE,}
        Elvin~Isufi,~\IEEEmembership{Student Member,~IEEE,}
        and~Geert~Leus,~\IEEEmembership{Fellow,~IEEE,}
\thanks{The authors are with the Faculty of Electrical Engineering, Mathematics and Computer Science, Department of Microelectronics, Delft University of Technology, 2628 CD Delft, The Netherlands. E-mails: \{j.liu-1, e.isufi-1, g.j.t.leus\}@tudelft.nl.  }
\thanks{ J.Liu is funded by the China Scholarship Council. This work is supported by the Circuits and Systems group, Delft University of Technology, The Netherlands.}
}
\maketitle

\begin{abstract}
In the field of signal processing on graphs, graph filters play a crucial role in processing the spectrum of graph signals. This paper proposes two different strategies for designing autoregressive moving average (ARMA) graph filters on both directed and undirected graphs. The first approach is inspired by Prony's method, which considers a modified error between the modeled and the desired frequency response. The second technique is based on an iterative approach, which finds the filter coefficients by iteratively minimizing the true error (instead of the modified error) between the modeled and the desired frequency response. The performance of the proposed algorithms is evaluated and compared with finite impulse response (FIR) graph filters, on both synthetic and real data. The obtained results show that ARMA filters outperform FIR filters in terms of approximation accuracy and they are suitable for graph signal interpolation, compression, and prediction.
\end{abstract}

\begin{IEEEkeywords}
Signal processing on graphs, autoregressive moving average graph filters, iterative processing, Prony's method.
\end{IEEEkeywords}

\IEEEpeerreviewmaketitle

\section{Introduction}

\IEEEPARstart{G}{raph} signal processing (GSP) extends classical digital signal processing to signals that live on the vertices of irregular graphs \cite{D, AJM}.
Similar to the frequency analysis of temporal signals, the definition of a Fourier-like transform for graph signals \cite{AJ} is a handle to process these signals in the so-called graph frequency domain, rather than only in the vertex domain \cite{ASJ2}. In this analogy, the frequency components of the graph signal characterize, now, the signal variation over the graph. The graph Fourier transform (GFT) has been defined in two ways, i.e., the projection of the graph signal onto the graph Laplacian eigenspace, see e.g., \cite{D}, or onto the eigenspace of the adjacency matrix, see e.g., \cite{ASJ2}. The first approach suits better undirected graphs characterized by \emph{real-valued} graph frequencies, whilst the second approach is preferred for directed graphs characterized by \emph{complex-valued} graph frequencies. Note that instead of the graph Laplacian or adjacency matrix, other so-called graph shift operators can also be considered~\cite{chung1997spectral}.

Together with the GFT, graph filters are a key tool to process the graph signal spectrum, i.e., to amplify or attenuate different graph frequencies. Graph filters find applications in graph signal denoising \cite{S, R3, R4}, smoothing \cite{zhang2008graph}, classification \cite{AS}, sampling \cite{R6}, recovery \cite{SC1} and graph clustering \cite{TPGV}. Further, they serve as a basic building block for trilateral graph filters \cite{R4}, graph filter banks \cite{AY,R7} and graph wavelets \cite{R1, DK, R2, R5}.
Finite impulse response (FIR) graph filters \cite{shuman2011distributed,ASJ2,segarra2017optimal,Coutinio2017}, direct analogs of temporal FIR filters, are implemented as a polynomial in the graph shift operator, e.g., the graph Laplacian matrix \cite{D}, the adjacency matrix~\cite{AJ}, or any modification thereof. To accurately match some prescribed specifications in the graph frequency domain, FIR filters require a high filter order leading to a high implementation cost. Furthermore, being matrix polynomials of the graph shift operator, their accuracy remains limited. This issue is especially present when the desired graph frequency response is characterized by sharp transitions, e.g., a step function, which is often required in clustering, graph filter banks, or to discriminate patients with different levels of adaptability in brain networks~\cite{LWNT}.

FIR filter design is already a well-established theory. One of the most popular approaches to fit the graph frequency response of the FIR filter to a desired spectrum is through solving a linear least squares (LLS) fitting problem~\cite{ASJ2}, which can be carried out for undirected as well as directed graphs. However, since the graph (and thus the set of graph frequencies) is not always perfectly known, techniques have been established to design the FIR filter coefficients without the knowledge of the graph spectrum, by fitting the frequency response over a continuous range of graph frequencies (we call this the universal design approach). The Chebyshev polynomial technique is a popular method in this context, but has only been introduced for undirected graphs in~\cite{shuman2011distributed}.
In this paper, we extend the LLS approach to a universal design method by gridding not only the real line (for undirected graphs), but also the complex plane (for directed graphs) and by subsequently solving the LLS problem for the graph frequencies that are on the grid.

An alternative to FIR graph filters are infinite impulse response (IIR) graph filters, such as the autoregressive moving average (ARMA) graph filters~\cite{EA}, or the gradient descent IIR graph filters~\cite{XM}. These filters are characterized by a rational graph frequency response, which brings more degrees of freedom to the design. However, the aforementioned works focus on a distributed implementation, which only leads to the modeled frequency response after an infinite number of iterations~\cite{EA,XM}. Further, the distributed implementation limits the filter approximation accuracy due to convergence constraints.

To fully exploit the benefits of the rational frequency response, in this paper, we focus on a centralized ARMA filter implementation. In a centralized fashion, the ARMA output can be simply found by solving a linear system of equations, which can be carried out efficiently with first order methods~\cite{PCVX} or conjugate gradient (CG)~\cite{SJA}. Based on this centralized implementation, we also propose new ARMA graph filter design methods, which can be adopted when the graph is known or in a universal fashion by gridding the frequency domain (as done for the LLS FIR filter design). The proposed ARMA design and implementation methods work for undirected as well as directed graphs. Throughout this work, we will mainly use FIR filters as a benchmark to assess the performance of the proposed ARMA filters, being their direct competitors, and propose ARMA filters as an alternative for the aforementioned applications.

The paper contribution is threefold:\vskip.1cm
\emph{i) We extend the universal LLS strategy to design FIR graph filters from undirected to directed graphs.} For either the normalized Laplacian (undirected graph) or normalized adjacency (directed graph) matrix, we respectively sample the real interval from zero to two or the complex unit disc. The first is done uniformly, whereas the second is done uniformly in amplitude and phase such that the obtained graph frequencies either appear in complex conjugate pairs or are purely real-valued. After the grid points have been determined, LLS is used to fit the response on these grid points.
\vskip.1cm
\emph{ii) We present an efficient centralized ARMA filter implementation.} ARMA filtering of graph signals is written as a linear system of equations, which can be solved by efficient of-the-shelf algorithms, such as CG \cite{SJA}. We propose the details of this implementation algorithm and present some simulation results.
\vskip.1cm
\emph{iii) We propose two ARMA graph filter design strategies, which can be applied to both directed and undirected graphs.} The first one is inspired by Prony's method \cite{M}, where a modified error between the modeled and the desired frequency response is minimized.
Meanwhile, the second approach minimizes the true error iteratively following the Steigliz-McBride idea~\cite{M}. As an initial condition, we can use the solution from the first method, thereby potentially improving the approximation accuracy of that solution. The two proposed methods can also be extended to a universal design by gridding the graph frequency domain as mentioned earlier.

Several numerical tests validate our findings with both synthetic and real data. We show that the ARMA filters outperform FIR filters in terms of approximation accuracy, even with fewer filter coefficients. In our tests with the real Molene temperature dataset, ARMA graph filters are used for interpolation purposes (on an undirected graph). With the same dataset, ARMA filters are also utilized to compress (on a directed graph) and predict (on both a directed and undirected graph) the graph signal. The results show that the error resulting from our ARMA filter design is lower than that resulting error from an FIR filter with the same number of filter coefficients.

{\it Paper outline}: Section~\ref{Sec.Pre} reviews some basic concepts of signal processing on graphs and introduces the concept of universal graph filter design. In Section~\ref{Sec.ARMA}, we introduce ARMA graph filtering, and the related ARMA　filter implementation. Section~\ref{sec.ARMAdes} contains the filter design problem and the proposed design strategies, while the simulation results are shown in Section~\ref{Sec.Numerical data}. Finally, the conclusions are drawn in Section~\ref{Sec.Concl}.

{\it Notation}: We indicate by normal letters $a$ or $A$ a scalar variable; a bold lowercase letter $\bf{a}$ will represent a vector variable and a bold uppercase letter $\bf{A}$ a matrix variable. We indicate the absolute value of $a$ by $|a|$ and the 2-norm of the vector $\bf{a}$ and matrix $\bf{A}$ by $\left\|\bf{a}\right\|$ and $\left\|\bf{A}\right\|$, respectively. $a_i$ or $[{\bf a}]_i$ represents the $i$-th entry of $\bf{a}$, and similarly $A_{i,j}$ or $[{\bf A}]_{i,j}$ represents the $(i,j)$-th entry of $\bf{A}$. ${\bf{a}}^{(i)}$ will indicate the value of $\bf{a}$ after the $i$-th iteration. ${{\bf{A}}^\dag }$ represents the pseudo-inverse of the matrix ${\bf{A}}$. We use ``$\circ$'' to represent the element-wise Hadamard product. We indicate the transpose and Hermitian of the matrix $\bf{A}$ by $\bf{A}^{\rm{T}}$ and ${\bf A}^{\rm H}$, respectively. The complex conjugate of $\bf{a}$ and ${\bf{A}}$ are represented as ${\bf{a}}^{*}$ and ${\bf{A}}^{*}$, respectively.

\section{Problem statement}
\label{Sec.Pre}

This section recalls some background information that will be used throughout the paper. We start with some preliminaries about GSP and graph filtering. Then, we formulate the general problem of graph filter design for some prescribed spectral requirements on both undirected and directed graphs. The notions of universal design and a review of the challenges in designing FIR graph filters conclude the section.

\subsection{Preliminaries}

Consider a graph $\mathcal{G = (V, E)}$ with $\mathcal{V}$ the set of $N$ nodes (vertices) and $\mathcal{E}$ the set of $E$ edges. The local structure of $\mathcal{G}$ is captured by the adjacency matrix ${\bf{A}} \in {\mathbb R}^{N \times N}$, where ${{[{\bf{A}}]_{j,i}} \ne 0}$ if there exists an edge between the nodes ${v_i}$ and ${v_j}$, or by the discrete graph Laplacian ${\bf L}_{\text{d}} = {\bf D} - {\bf A} \in {\mathbb R}^{N \times N}$, where ${\bf D}$ is the diagonal degree matrix with diagonal entries defined as $[{\bf D}]_{i,i} = \sum_{j=1}^N [{\bf A}]_{i,j}$ (in-degree matrix) or $[{\bf D}]_{i,i}= \sum_{j=1}^N [{\bf A}]_{j,i}$ (out-degree matrix). Note that for an undirected graph $\mathcal{G}$, every edge between ${v_i}$ and ${v_j}$ leads to a similar edge between ${v_j}$ and ${v_i}$, and thus ${\bf A}$ is symmetric, i.e., ${{[{\bf{A}}]_{i,j}} = {[{\bf{A}}]_{j,i}}}$. This means that also the discrete graph Laplacian ${\bf L}_{\text{d}}$ is symmetric and there is no difference between using the in-degree or out-degree matrix. For directed graphs $\mathcal{G}$, such properties do not hold.

Throughout this paper we will use the adjacency matrix $\bf{A}$ as a representative for directed graphs, while for undirected graphs we use as alternative the discrete graph Laplacian ${\bf{L}_\text{d}}={\bf{D}}-{\bf{A}}$. More specifically, we will use their normalized counterparts, i.e., the normalized adjacency matrix ${\bf A}_{\text{n}} = {\bf A} / \| {\bf A} \|$ for directed graphs and the normalized Laplacian matrix ${\bf{L}}_\text{n}{\rm{ = }}{{\bf{D}}^{-{1 \mathord{\left/ {\vphantom {1 2}} \right. \kern-\nulldelimiterspace} 2}}}{\bf{L}_\text{d}}{{\bf{D}}^{-{1 \mathord{\left/ {\vphantom {1 2}} \right. \kern-\nulldelimiterspace} 2}}}$, for undirected graphs. Note that other alternatives can also be used. In short, every one of these graph representations can be referred to as a so-called graph shift operator $\bf{S}$, an operator that forms the basis for processing graph signals, as we will see next.
In this paper, we restrict ourselves to graphs for which ${\bf{S}}$ is real-valued and diagonalizable, and thus enjoys an eigenvalue decomposition $ {\bf S} = {\bf{U\Lambda}} {{\bf{U}}^{\rm {-1}}}$, with ${{\bf{U}}}$ the eigenvector matrix containing as columns the so-called graph modes ${\bf u}_1$ up to ${\bf u}_N$ and ${\bf{\Lambda}}$ the diagonal eigenvalue matrix containing as diagonal entries the so-called graph frequencies $\lambda_1$ up to $\lambda_N$ (note in this context that $\| {\bf S} \| = \max_n | \lambda_n | = | \lambda_{\max} |$).

For undirected graphs, ${\bf{S}}$ is symmetric and normal. The graph frequencies are in this case real-valued, and since ${\bf S}$ is real-valued, the graph modes are assumed to be real-valued as well (note that in some cases, they can be chosen to be complex-valued, e.g., for undirected circulant graphs, but that is not assumed in this paper). Specifically, for an undirected graph with ${\bf{S}} = {\bf{L}}_\text{n}$, the graph frequencies are in the real interval from zero to two. They can be ordered from small to large, where a smaller value indicates a lower frequency~\cite{D}.

For directed graphs, the graph frequencies are complex-valued but since ${\bf S}$ is real-valued they either appear in complex conjugate pairs or are purely real-valued. Moreover, the related graph modes also appear in complex conjugate pairs or are purely real-valued. Specifically, for the shift operator ${\bf{S}} = {\bf{A}}_\text{n}$, the graph frequencies are in the complex unit disc.
They can be ordered by the graph total variation of the related graph modes, which is defined as $\text{TV}_{{\mathcal G}} ({\bf{u}}_{n}) = \left| {1 - {\lambda_{n}  \mathord{\left/ {\vphantom {\lambda  {\left| {{\lambda _{\max }}} \right|}}} \right. \kern-\nulldelimiterspace} {\left| {{\lambda _{\max }}} \right|}}} \right|{\left\| {{\bf{u}}_{n}} \right\|_1}$. In other words, the frequencies are ordered according to the similarity between the $n$th graph mode and its graph shifted version.
Graph frequencies closer to the point $(1,0)$ in the complex plane will represent lower frequencies in this context~\cite{ASJ2}. See Fig. 1 for an example of a directed graph and its complex-valued graph frequencies.

We will indicate with the vector ${\bf{x}} \in {\mathbb R}^{N \times 1}$ the real-valued graph signal, i.e., a signal living on the nodes of the graph $\mathcal{G }$, where each value ${x_i}$ is associated to the node ${v_i}$. To obtain the graph frequency representation of ${\bf{x}}$, the eigenvector matrix ${\bf{U}}$ is used to transform the signal into the graph Fourier domain. Specifically, the GFT ${\hat{\bf x}}$ of ${\bf{x}}$ and its inverse are, respectively, ${{\hat{\bf x}}={\bf U}^{-1}{\bf{x}} }$ and ${{ {\bf x}}={\bf U}{\hat {\bf{x}}} }$. The following property can now be stated.

 \begin{property}\label{prop:GFT}
For either an undirected or directed graph ${\mathcal G}$, let us denote $\hat{x}_n$ as the $n$th frequency coefficient of the graph signal ${\bf x}$.
Then, the frequency coefficient $\hat{x}_n$ related to the real-valued graph frequency (mode) $\lambda_n$ (${\bf u}_n$) is real-valued as well. Meanwhile, the frequency coefficients $\hat{x}_n$ and $\hat{x}_{n'}$ related to the complex conjugate pair of graph frequencies (modes) $\lambda_n$ and $\lambda_{n'}$ (${\bf u}_n$ and ${\bf u}_{n'}$) form a complex conjugate pair as well.
\end{property}

This property is built on the fact that for a real-valued matrix ${\bf S}$, eigenvalues and eigenvectors appear in complex conjugate pairs~\cite{edwards2004differential, sinswat1977eigenvalue}. This also means that if the columns ${\bf u}_n$ and ${\bf u}_{n'}$ in the matrix ${\bf U}$ form a complex conjugate pair, the related rows in the matrix ${\bf U}^{-1}$ form a complex conjugate pair. Thus, with ${\bf U}^{-1}{\bf{x}}$, the frequency coefficients $\hat{x}_n$ and $\hat{x}_{n'}$ appear as a complex conjugate pair.

\begin{figure}[!t]
  \centering
  \includegraphics[trim={1.5cm 0 2cm 0},clip,width=.48\textwidth]{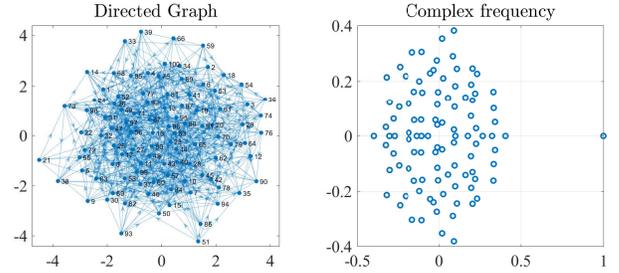}
  \caption{Directed graph of $N=100$ nodes with $E = 752$ edges having different weights in the interval $[0, 3]$. Complex-valued frequencies are generated by the eigenvalue decomposition of the normalized adjacency matrix ${\bf{A}}_{\text{n}}$. The ``largest'' frequency has magnitude one. Some frequencies live on the real axis while the remaining frequencies appear as conjugate pairs in the complex plane. }
  \label{fig:dgcpx}
\end{figure}

For a more in-depth analysis on the basics of the GFT and the ordering of graph frequencies we redirect the reader to \cite{D},  \cite{AJ}, and \cite{ASJ2}.

\subsection{Graph filtering}\label{sec:filtering}

A graph filter ${\bf{G}}$ is a function $g( \cdot )$ applied to the shift operator $\bf{S}$, i.e., ${\bf{G}} ={{g}{({\bf S} )}}$, that allows an eigendecomposition of $\bf{G}$ in the form $ {\bf G} = {\bf U}g({\bf \Lambda}){\bf U}^{-1}$, where $g({\bf \Lambda})$ is a diagonal matrix that highlights the filter impact on the graph frequencies ${\bf \Lambda}$. More specifically, the filter output ${\bf y}$ for a filter input ${\bf x}$ can be written as ${\bf y}={\bf G}{\bf x}$, which in the graph frequency domain can be translated into ${\hat{\bf y} = g( {\bf \Lambda} ) \hat{\bf x}}$, where $\hat{\bf x}$ and $\hat{\bf y}$ represent the GFT of the input and output signal, respectively. Hence, ${ g( {\bf \Lambda} ) }$ has on the diagonal the frequency response of the filter, which at frequency $\lambda_n$ we denote as $[  g( {\bf \Lambda} ) ]_{n,n} = \hat{g}_n$.

Throughout this paper, we will consider different parametrizations of the graph filter function $g(\cdot)$, and thus we will often explicitly write this function as $g(\cdot ; {\boldsymbol \theta})$, where ${\boldsymbol \theta}$ is a vector that contains the graph filter parameters, i.e., filter coefficients, zeros and poles, or any other set of filter parameters. Correspondingly, we can also write $\hat{g}_n$ explicitly as $\hat{g}_n({\boldsymbol \theta})$. Assuming now that the desired frequency response at frequency ${\lambda_n}$ is given by ${{\hat h}_n}$, the filter parameters ${\boldsymbol \theta}$ can be found by solving
\begin{equation}\label{eq:FilG1}
\mathop {\min }\limits_{\boldsymbol \theta} \sum_{n=1}^N | {{{\hat h}_n} - {{\hat g}_n}( {\boldsymbol \theta} )} |^2  .
\end{equation}
The desired frequency response $\hat{h}_n$ can originate from different scenarios. For instance, when we focus on graph filter design, i.e., when we want to design a low pass filter to smooth or denoise a graph signal, the desired frequency response $\hat{h}_n$ basically indicates how much we want to attenuate a specific graph mode and thus it will generally be real-valued and symmetric w.r.t. the real axis (for both undirected and directed graphs). Also when we want to do graph signal prediction, as done in~\cite{AJ}, we basically want to design an all-pass filter and set $\hat{h}_n$ to be one (and thus real-valued) everywhere. In this case, the cost function~\eqref{eq:FilG1} will also be weighted, as we will show in the simulations, but the filter design methods that we derive later on can easily be adapted to this weighting. However, for some GSP applications, such as compression, the desired frequency response $\hat{h}_n$ will be the GFT of the signal, for which Property~\ref{prop:GFT} holds.

In any case, whatever the scenario (filter design, prediction, smoothing, denoising, or compression) or type of graph (undirected or directed), the following property holds.

\begin{property}\label{prop:design}
As mentioned above, $\hat{h}_n$ is selected either as real-valued and symmetric w.r.t. the real frequency axis, or as the GFT of a signal. The latter means that $\hat{h}_n$ is real-valued if $\lambda_n$ is real-valued while $\hat{h}_n$ and $\hat{h}_{n'}$ form a complex conjugate pair if $\lambda_n$ and $\lambda_{n'}$ form a complex conjugate pair (this is due to Property~1). Put differently, either way we select $\hat{h}_n$, if $\lambda_n$ is real-valued, then $\hat{h}_n$ is real-valued whereas if $\lambda_n$ and $\lambda_{n'}$ form a complex conjugate pair, then $\hat{h}_n$ and $\hat{h}_{n'}$ form a complex conjugate pair as well.
\end{property}

\subsection{Universal filter design}\label{sec:universal}

\begin{figure}[!t]\label{fig:lowpass}
\centering
\subfloat[]{\includegraphics[trim={1.5cm 0 2cm 0},clip,width=.48\textwidth]{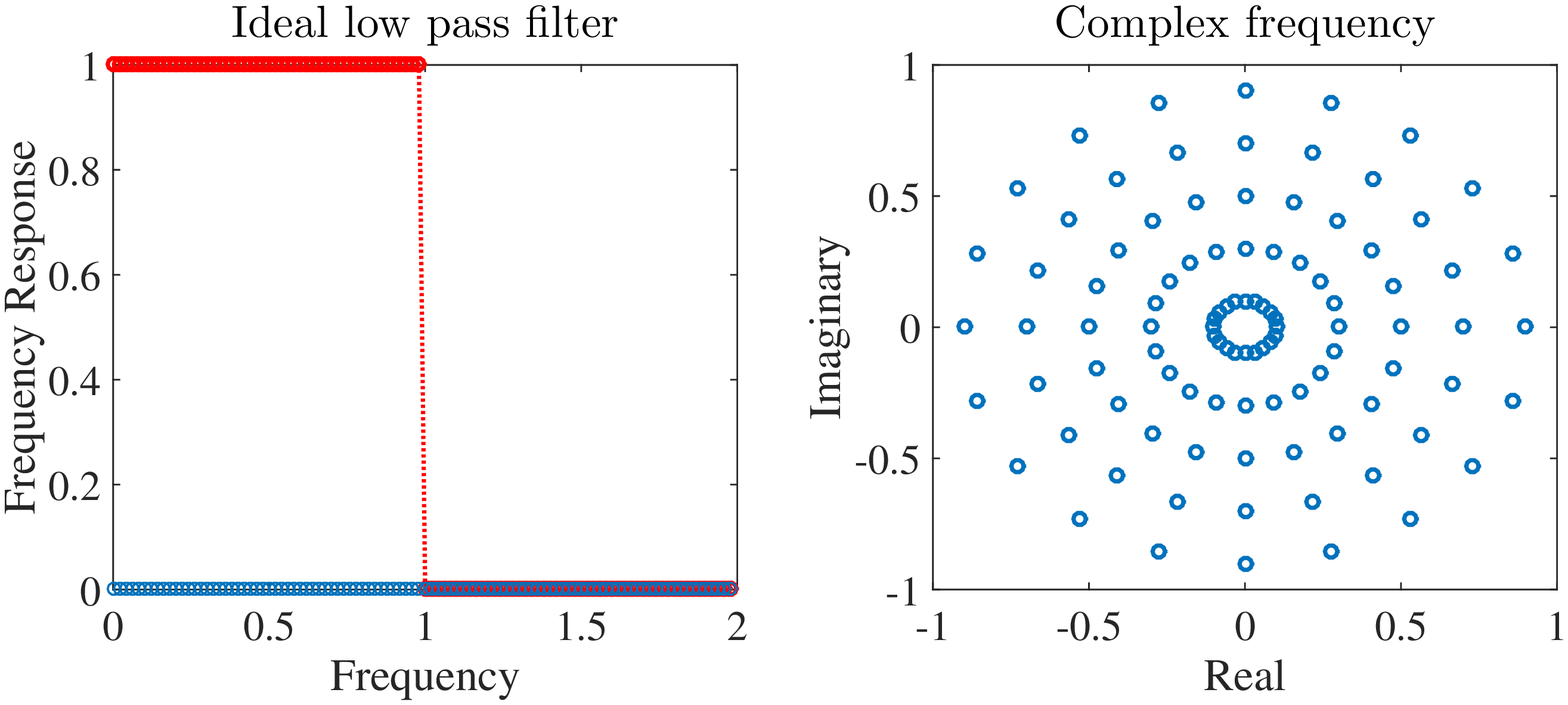}%
\label{subfig1}}\\
\subfloat[]{\includegraphics[trim={1.5cm 0 2cm 0},clip,width=.48\textwidth]{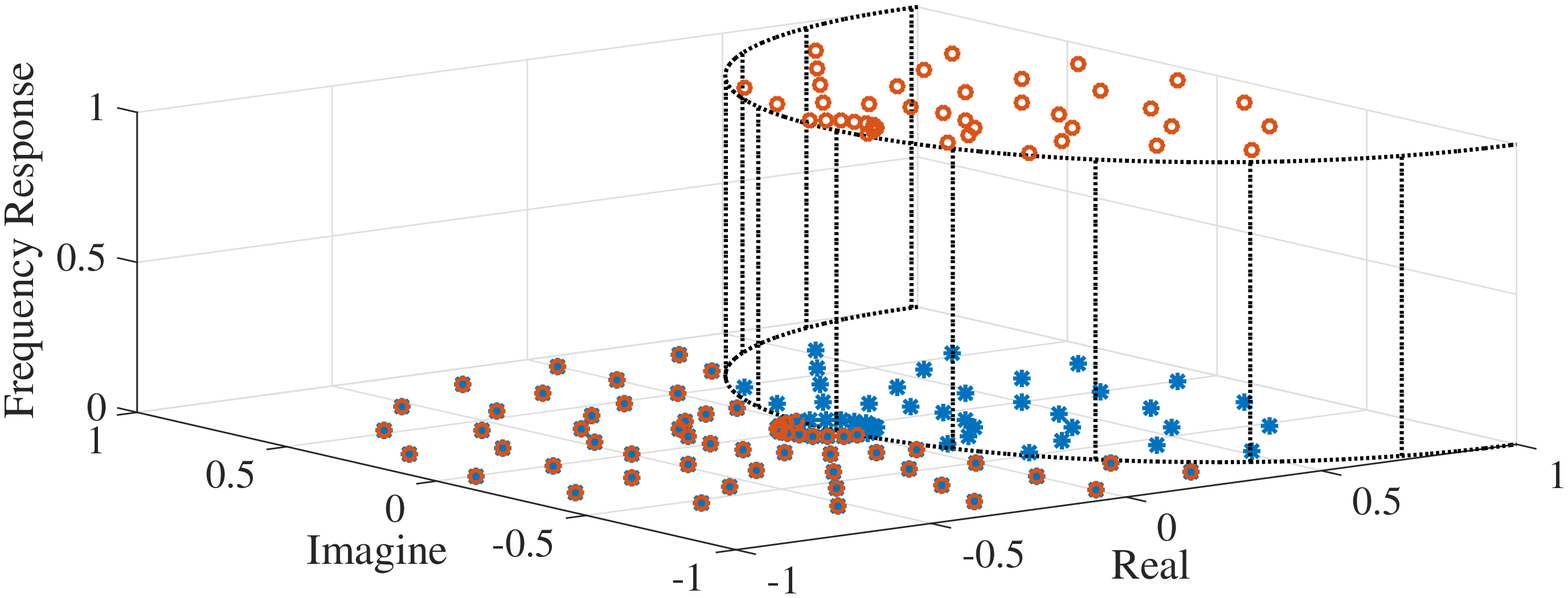}%
\label{subfig1}}
\caption{(a) (left) Ideal low pass filter response of universal design for undirected graph ($N=100$). (a) (right) Universal gridding for directed graph associated with the normalized adjacency matrix  ${\bf{A}}_{\text{n}}$ ($N=100$). (b) Ideal low pass filter response of universal design for directed graph with $N=100$. The complex frequencies lying inside the circle with radius 1 centered at $(1, 0)$ are ''small'' frequencies.}
\label{fig:filg}
\end{figure}

Since estimating the graph frequencies entails some additional complexity, graph filters are often designed with no explicit knowledge of the graph or the graph frequencies. The desired frequency response is assumed to be a function over a continuous range of frequencies (the real line for undirected graphs or the complex plane for directed graphs). Solving the filter design problem for such a scenario is referred to as \emph{universal filter design}. Following the same LLS approach as in~\eqref{eq:FilG1}, this universal filter design problem can be tackled by discretizing the related continuous frequency range into a finite set of graph frequencies. Then problem~\eqref{eq:FilG1} can be solved for this finite set of graph frequencies instead of for the true graph frequencies.

For undirected graphs with ${\bf{S}} = \bf{L}_\text{n}$, we can consider for instance $N$ different grid points in the interval $[0, 2]$. Note that depending on the graph, we obtain a different eigenvalue spread, e.g., the eigenvalues of an Erd\H{o}s  R\'{e}nyi graph~\cite{PA} are in general closely spread around 0 and more widely spread around $p$, the link probability of the graph (see Fig.~\ref{fig:simMP} (f)). However, since we want to be independent of the graph topology, we consider a uniformly-spaced grid in our design. As an example, we show the graph spectrum for an ideal low pass graph filter with cutoff frequency $\lambda_c = 1$ in Fig.~\ref{fig:filg}(a) left.

Alternatively, for directed graphs with ${\bf{S}} = \bf{A}_\text{n}$, the graph frequencies lie in the complex unit disc. Again trying to avoid any dependence on the graph, we suggest gridding this disc by $N$ complex conjugate pairs of points, as shown in Fig.~\ref{fig:filg}(a) right. Fig.~\ref{fig:filg}(b) again shows an example of an ideal low pass filter in this context. The cutoff frequency $\lambda_c$ is here defined as the distance from the point $(1,0)$ in the complex plane, and it is set as $\lambda_c=1$ in Fig.~\ref{fig:filg}(b). All graph frequencies with a distance to $(1,0)$ that is smaller than $\lambda_c$ will be part of the passband since they yield the ``smaller'' frequencies.

\subsection{FIR graph filters}

From \cite{AJ}, an FIR graph filter ${\bf{G}}$ of order $K$ can be expressed as a $K$-th order polynomial in the graph shift operator
\begin{equation}\label{eq:FIRK}
{\bf{G}} ={{g}{({\bf S}; {\boldsymbol \theta})} } =  \sum\limits_{k = 0}^{K} g_k {\bf S}^k,
\end{equation}
with ${\boldsymbol \theta} = {\bf g} = [ g_0, \dots, g_K]^T$ collecting the FIR filter coefficients.
The filter frequency response at frequency ${\lambda_n}$ can be expressed as
\begin{equation}\label{eq:FIRlambda}
\hat{g}_n = \sum\limits_{k = 0}^{K} g_k \lambda_n^k.
\end{equation}
By stacking the filter frequency response in ${\hat {\bf{g}} = [{{\hat g}_1}, \cdots ,{{\hat g}_N}]^{\rm T}}$, we obtain the relation
\begin{equation}\label{eq4}
\hat {\bf g} = {\bf{\Psi}}_{K+1} {\bf{g}},
\end{equation}
where ${{\bf{\Psi }}_{K+1}}$ is the $N \times (K+1)$ Vandermonde matrix with entries $[{\bf \Psi}]_{n,k} = \lambda_n^{k-1}$.
Assuming the desired frequency response is given by ${\hat{\bf h}} = [ \hat h_1, \cdots , \hat h_N]^{\rm T}$, problem~\eqref{eq:FilG1} can now be written as the following linear least squares (LLS) problem
\begin{equation}\label{eq:FIRproblem}
\min_{{\bf g}} \| \hat{\bf h} -{\bf{\Psi}}_{K+1} {\bf{g}} \|^2.
\end{equation}
The solution of this LLS problem is given by
\begin{equation}\label{FIRpinv}
{\bf{g}} = {\bf{\Psi }}_{K + 1}^\dag \hat {\bf{h}},
\end{equation}
where ${{\bf{\Psi }}_{K + 1}^\dag}$ is the pseudo-inverse of ${{\bf{\Psi }}_{K + 1}}$. As shown in~\cite{ASJ2}, \cite{segarra2017optimal}, ${{\bf{\Psi }}_{K + 1}}$ needs to be well-conditioned for this approach to work well. This will only be the case for small graph sizes $N$ and/or small filter orders $K$. Note that to improve the conditioning, close eigenvalues could be grouped together under the assumption that the desired filter response on those eigenvalues is equal. In any case, the FIR filter order $K$ needs to be small and because of the nature of the polynomial fitting problem, this will lead to a limited accuracy of the FIR filter.

For~\eqref{eq:FIRlambda} to make sense as a graph filter that will be applied to a real-valued graph signal ${\bf x}$, we want the FIR filter coefficients ${\bf g}$ to be real-valued. The next Proposition shows that this is the case.
\begin{proposition}\label{th:1}
Under Property~\ref{prop:design}, the FIR filter coefficients ${\bf g}$ obtained by solving~\eqref{eq:FIRproblem} are real-valued.
\end{proposition}

\begin{proof}
The proof can be found in Appendix A.
\end{proof}

\section{ARMA graph filter and implementation}
\label{Sec.ARMA}

To improve the approximation accuracy and reduce the number of required filter coefficients w.r.t.~the FIR filter, we now consider applying an ARMA filter to the graph signal $\bf x$.
In this section, we first introduce the ARMA graph filtering problem. Then, a centralized ARMA filter implementation is presented, and some issues related to the corresponding filter design problem are highlighted. Solutions to this ARMA filter design problem are presented in Section~\ref{sec.ARMAdes}.

\subsection{ARMA graph filter}

From~\cite{EA}, and similar to temporal ARMA filters~\cite{M}, an ARMA graph filter is characterized by a rational polynomial in the graph shift operator
\begin{equation}\label{eq:ARMAvertex}
  {\bf G} = g({\bf S} ; {\boldsymbol \theta}) = \left( \sum_{p=0}^P a_p {\bf S}^p \right)^{-1} \sum_{q=0}^Q b_q {\bf S}^q,
\end{equation}
where ${\boldsymbol \theta} = [ {\bf a}^T, {\bf b}^T ]^T$ with ${\bf a} = [ a_0, \dots, a_P]^T$ and ${\bf b}= [b_0, \dots, b_Q ]^T$ collecting the ARMA filter coefficients. This allows us to express the filter frequency response at frequency $\lambda_n$ as
\begin{equation}\label{eq:ARMAfunction}
\hat {g}_n = \frac{\sum\nolimits_{q=0}^Q b_q \lambda_n^q}{\sum\nolimits_{p=0}^P a_p \lambda_n^p }.
\end{equation}
Stable ARMA filters are obtained when $\sum_{p=0}^P a_p {\bf S}^p$ is invertible, or equivalently, when $\sum_{p=0}^P a_p \lambda_n^p$ is different from zero for all $n=1,2,\dots,N$.
This stability condition is less critical as in the time domain, which is mainly due to the fact that a graph signal is finite-length whereas a temporal signal is infinite-length. Hence, there is no big risk of the filter output growing unbounded.

Note that for simplicity reasons we define the ARMA filter coefficients ${\bf a}$ and ${\bf b}$ in an ambiguous way since multiplying both ${\bf a}$ and ${\bf b}$ with the same constant will not change the ARMA graph filter. Hence, whenever we design ${\bf a}$ and ${\bf b}$, we will remove this ambiguity by constraining the first AR coefficient to be one, i.e., $a_0=1$, which is rather standard.

\subsection{Implementation of ARMA graph filter}

From~\eqref{eq:ARMAvertex}, it is clear that the relation between the output ${\bf y}$ and the input ${\bf x}$ of an ARMA graph filter is given by
\begin{equation}\label{L2}
{ \left(  \sum\limits_{p = 0}^P {{a_p}{{\bf{S}}^p}} \right)}{\bf{y}} = \left( \sum\limits_{q = 0}^Q {{b_q}{{\bf{S}}^q}} \right) {\bf{x}}.
\end{equation}
Hence, by defining the matrices
\begin{equation}\label{eq:predPQ}
{\bf{P}} = \sum\limits_{p = 0}^P {{a_p}{{\bf{S}}^p}}, \quad \quad {\bf{Q}} = \sum\limits_{q = 0}^Q {{b_q}{{\bf{S}}^q}},
\end{equation}
we can express~\eqref{L2} in the compact form
\begin{equation}\label{L3}
{\bf{P}} {\bf{y}} = {\bf{Q}} {\bf{x}}.
\end{equation}
To compute the filter output ${\bf{y}}$ in \eqref{L3}, we can first calculate the right-hand side denoted for commodity as ${{\bf{z}}} = {\bf{Qx}}$ (which corresponds to pre-filtering $\bf{x}$ with an FIR filter) and then ${\bf{y}}$ is found by simply solving the linear system
\begin{equation}\label{L4}
{\bf{P}} {\bf{y}} = {{\bf{z}}}.
\end{equation}
Note that there are several efficient methods to solve \eqref{L4}, like first order methods \cite{PCVX}, the power method \cite{wilkinson1965algebraic}, and conjugate gradient (CG) \cite{SJA}. Their computational cost reduces significantly for sparse matrices ${\bf S}$, i.e., for sparse graphs \cite{M2}.

In this work we consider the CG method \cite{SJA} to implement ARMA graph filters in the vertex domain. As shown in Algorithm 1, the CG approach has a computational complexity that scales linearly in the number of edges $E$. Specifically, we first need to compute $\bf{z}= {\bf{Q}}\bf{x}$, which by following the efficient implementation \cite{DPP} requires $Q$ multiplications with the shift operator ${\bf S}$ since the terms can be computed as ${\bf S}^k{\bf x} = {\bf S}({\bf S}^{k-1}{\bf x})$ leading to an overall complexity of $O(QE)$. Then, in each iteration $i$ of the CG it is required to compute the term ${\bf P}{\bf d}^{(i)}$, which if computed in the same way as $\bf{z}$ requires a computational effort of order $O(PE)$. Thus, if considering that the CG is arrested after $T$ iterations, the overall implementation cost of the ARMA graph filter is of order $O((PT + Q)E)$. We would like to highlight that the ARMA filter output with CG is computed without explicitly building the matrices ${\bf P}$ and ${\bf Q}$, and only considering their application to a specific vector.

\begin{table}[!t]\renewcommand{\arraystretch}{1.2}
\centering
\begin{tabular}{rrl}
     \hline
     \hline
      \multicolumn{3}{l}{ $\bf Algorithm 1$: Conjugate gradient} \\
       \hline
       1&$\bf Input$: & ${{{\bf{y}}^{(0)}}}$, ${{\bf{x}}}$, coefficients ${a_p}$, $ {b_q}$ \\
       2& & accuracy ${\varepsilon}$, number of iterations ${T}$\\
       3&$\bf Initialization$: & ${{\bf{z}}}$, ${\bf{P}}{{\bf{y}}^{(0)}}$ (using ${{\bf{S}}^k}{{\bf{y}}^{(0)}} = {\bf{S}}({{\bf{S}}^{k - 1}}{{\bf{y}}^{(0)}})$) \\
       4& & ${{\bf{d}}^{(0)}} = {{\bf{r}}^{(0)}} = {{\bf{z}}} - {\bf{P}}{{\bf{y}}^{(0)}}$,\\
       5& & ${\delta ^{(0)}} = {\delta ^{new}} = {\bf{r}}^{(0)T}{{\bf{r}}^{(0)}}$ \\
       6&$\bf Iteration$: & while $i < T$ and ${\delta ^{new}} > {\varepsilon ^2}{\delta ^{(0)}}$\\
       7& & ${\omega ^{(i)}} = \frac{{{\delta _{new}}}}{{{\bf{d}}^{(i)T}{\bf{P}}{{\bf{d}}^{(i)}}}}$ \\
       8& & ${{\bf{y}}^{(i + 1)}} = {{\bf{y}}^{(i)}} + {\omega^{(i)}}{{\bf{d}}^{(i)}}$, \\
       9& & ${{\bf{r}}^{(i + 1)}} = {{\bf{r}}^{(i)}} - {\omega^{(i)}}{\bf{P}}{{\bf{d}}^{(i)}}$ \\
       10& & ${\delta ^{old}} = {\delta ^{new}}$, ${\delta ^{new}} = {\bf{r}}^{(i + 1)T}{{\bf{r}}^{(i + 1)}}$\\
       11& & ${\varphi^{(i + 1)}} = \frac{{{\delta ^{new}}}}{{{\delta ^{old}}}}$, ${{\bf{d}}^{(i + 1)}} = {{\bf{r}}^{(i + 1)}} + {\varphi^{(i + 1)}}{{\bf{d}}^{(i)}}$ \\
       12& & $i=i+1$ \\
       13&$\bf Output$: & ${{\bf{y}}^{(i + 1)}}$ \\
       \hline
       \hline
   \end{tabular}
\end{table}

In Section \uppercase\expandafter{\romannumeral5}, we analyze the tradeoff between the computational implementation cost and approximation accuracy induced by the CG approach.

\section{ARMA graph filter design}
\label{sec.ARMAdes}
This section contains the proposed ARMA filter design methods. We start with a discussion of the ARMA design problem, followed by two approaches inspired by Prony's method, and finally an iterative approach.

\subsection{ARMA design problem}
\label{subsec.ARMAdesign}

As discussed in Subsection~\ref{sec:filtering}, we would like to find the ARMA filter coefficients ${\bf a}$ and ${\bf b}$ such that a desired frequency response $\hat{h}_n$ is matched, where the latter can be a desired filter shape (for filter design, smoothing, or denoising) or the GFT of a graph signal (for compression or prediction). In this context, note that many desired responses $\hat{h}_n$ already have the shape of an ARMA filter, e.g., for Tikhonov denoising or interpolation, which means no explicit fitting is required in that case.

More specifically, adapting~\eqref{eq:FilG1} to our ARMA filter design problem, we want to minimize the following error
\begin{equation}\label{eq:optimal}
e_n = {{\hat h}_n}- \frac{{\sum\nolimits_{q = 0}^Q {{b_q}\lambda _n^q} }}{{\sum\nolimits_{p = 0}^P {{a_p}\lambda _n^p} }}.
 \end{equation}
Since~\eqref{eq:optimal} is nonlinear in ${\bf a}$ and ${\bf b}$, classical approaches like Prony's method~\cite{M} consider minimizing the following modified error
\begin{equation}\label{eq:prony}
 e'_n = {{\hat h}_n}\left( {\sum\limits_{p = 0}^P {{a_p}} \lambda _n^p} \right) - \sum\limits_{q = 0}^Q {{b_q}} \lambda _n^q.
 \end{equation}
 The latter is clearly not equivalent to~\eqref{eq:optimal} but it is linear in ${\bf a}$ and ${\bf b}$.

In the sequel, our goal will be to find ${\bf a}$ and ${\bf b}$ that minimize~\eqref{eq:optimal} or~\eqref{eq:prony} in the mean square sense, subject to~$a_0=1$ as mentioned before. Similar to the FIR filter, if we want the ARMA filter to make sense as a graph filter that will be applied to a real-valued graph signal ${\bf x}$, we want the ARMA filter coefficients ${\bf a}$ and ${\bf b}$ to be real-valued. We will show that this is the case for the different proposed approaches. Finally, note that, similar to Prony's method~\cite{M}, the non-convex stability constraint $\sum\nolimits_{p = 0}^P {{a_p}} \lambda _n^p \ne 0$ will be ignored in the rest of the paper, but it can easily be checked after the design.

\subsection{Methods inspired by Prony}

\textbf{Prony's LS.} To start, let us first stack $e_n$ from~\eqref{eq:optimal} in the vector ${\bf e}=[e_1, \dots, e_N]^{\rm T}$, which can be expressed as
\begin{equation}\label{eq:original_error}
{\bf{e}} = \hat{\bf h} - \text{diag}( {\bf \Psi}_{P+1} {\bf a} )^{-1} {\bf \Psi}_{Q+1} {\bf b} .
\end{equation}
As we mentioned before, this nonlinear function is hard to handle and thus we focus on the modified error.
Stacking $e'_n$ from~\eqref{eq:prony} in the vector ${\bf e}'=[e'_1, \dots, e'_N]^{\rm T}$, we obtain the simpler linear expression
\begin{align}\label{eq:other}
{{\bf{e}}' } & = \hat{\bf h} \circ ({{\bf{\Psi }}_{P + 1}}{\bf{a}}) - {{\bf{\Psi }}_{Q + 1}}{\bf{b}} \\ & = [{{\bf{\Psi }}_{P + 1}} \circ ( \hat{\bf h} {\bf 1}_{P+1}^{\rm{T}} )]{\bf{a}} - {{\bf{\Psi }}_{Q + 1}}{\bf{b}},
\end{align}
where ``$\circ$'' represents the element-wise Hadamard product and ${\bf 1}_{P+1}$ is the $(P+1) \times 1$ all-one vector.

Minimizing $\| {\bf e}' \|^2$ over ${\bf a}$ and ${\bf b}$ leads to the following LLS problem
\begin{equation}\label{eq:LSsolution}
\mathop {\min }\limits_{{\bf{a}},{\bf{b}}} {\left\| [ {{\bf{\Psi }}_{P + 1}} \circ ({\hat {\bf h}}{{\bf 1}_{P + 1}^{\rm{T}}), - {{\bf{\Psi }}_{Q + 1}}} ] \begin{bmatrix}
{\bf{a}} \\
{\bf{b}}
\end{bmatrix} \right\|^2}, \ \text{s.t.} \ {a_0} = 1,
\end{equation}
which can be solved efficiently. The next Proposition shows that the obtained ${\bf a}$ and ${\bf b}$ vectors are real-valued.
\begin{proposition}\label{th:2}
Under Property~\ref{prop:design}, the ARMA filter coefficients ${\bf a}$ and ${\bf b}$ obtained by solving~\eqref{eq:LSsolution} are real-valued.
\end{proposition}
\begin{proof}
The proof is similar to the proof of Proposition~\ref{th:1}.
\end{proof}

\textbf{Prony's projection.} Since Prony's LS approach addresses the modified error~\eqref{eq:prony} and not the desired error~\eqref{eq:optimal}, we here consider a way to partly overcome this limitation, and potentially improve the approximation accuracy of~\eqref{eq:LSsolution}. We use the orthogonal subspace projection approach~\cite{CCH} to rephrase~\eqref{eq:other} as a function of only the denominator coefficients ${\bf a}$. Then, with the obtained solution for ${\bf a}$, the original error~\eqref{eq:optimal} can be minimized to find the numerator coefficients ${\bf b}$. This approach can be interpreted as Shanks' method similar to that used in~\cite{EA}.

Let us start by considering the orthogonal projection matrix onto the orthogonal complement of the range of ${{\bf \Psi}_{Q+1}}$
\begin{equation}\label{eq:PM}
{\bf P}^{\perp}_{{\bf \Psi}_{Q+1}} = {\bf I}_N - {\bf \Psi}_{Q+1} {\bf \Psi}_{Q+1}^\dagger,
\end{equation}
where ${{\bf \Psi}_{Q+1}}$ is better conditioned than ${{\bf \Psi}_{K+1}}$ used to design an FIR graph filter, because $Q < K$ and removing columns from a tall matrix improves its condition number. Then, the modified error~\eqref{eq:other} can be reshaped as
\begin{equation}\label{eq:OSPerror}
{{\bf e}''} = {\bf{P}}_{{{\bf{\Psi }}_{Q + 1}}}^ \bot [{{\bf{\Psi }}_{P + 1}} \circ ( \hat{\bf h} {\bf 1}_{P+1}^{\rm{T}} )]{\bf{a}} - {\bf{P}}_{{{\bf{\Psi }}_{Q + 1}}}^ \bot {{\bf{\Psi }}_{Q + 1}}{\bf{b}},
\end{equation}
where the second term on the right hand side of~\eqref{eq:OSPerror} is zero. As shown in~\cite{hu2017unified},~\cite{CCH}, this projection operator preserves the solution for ${\bf a}$ when minimizing~\eqref{eq:OSPerror} instead of~\eqref{eq:other}. Hence, after the projection, the LLS problem for solving ${\bf a}$ becomes
\begin{equation}\label{eq:LS5}
\min_{{\bf a}} \|  {\bf P}^{\perp}_{{\bf \Psi}_{Q+1}} [ {\bf \Psi}_{P+1} \circ ( \hat{\bf h} {\bf 1}_{P+1}^{\rm T} ) ] {\bf a} \|^2, \
\text{s.t.} \ {a_0} = 1.
\end{equation}
The reason why we prefer solving~\eqref{eq:LS5} over~\eqref{eq:LSsolution} for finding a solution for ${\bf a}$ is the computational complexity. Finally, the vector ${\bf b}$ can be obtained using~\eqref{eq:optimal} after plugging in the solution for ${\bf a}$ obtained from~\eqref{eq:LS5}. In other words, ${\bf b}$ is found by solving
\begin{equation}\label{eq:LS6}
\min_{{\bf b}} \|  \hat{\bf h} - \text{diag}( {\bf \Psi}_{P+1} {\bf a} )^{-1}  {\bf \Psi}_{Q+1} {\bf b} \|^2.
\end{equation}
As before, we can again show that this solution for ${\bf a}$ and ${\bf b}$ is real-valued.
\begin{proposition}\label{th:3}
Under Property~\ref{prop:design}, the ARMA filter coefficients ${\bf a}$ and ${\bf b}$ obtained by solving~\eqref{eq:LS5} and~\eqref{eq:LS6} are real-valued.
\end{proposition}
\begin{proof}
The proof is similar to the proof of Proposition~\ref{th:1}.
\end{proof}

We would like to remark that this version of Prony's projection approach has a conceptual difference with the method presented in \cite{EA}. While in \cite{EA} the desired frequency response is first fitted  with an FIR filter and then the denominator coefficients are found to match that response, we here aim at approaching directly the desired response rather than its FIR approximation. In parallel to the classical literature \cite{M}, our approach can be considered as a reshaping of the Pad\'e approximation which first is solved for the denominator coefficients ${\bf a}$ and then for the numerator coefficients ${\bf b}$. As we show in Section \uppercase\expandafter{\romannumeral5}, the Prony's projection approach improves in general the approximation accuracy of \eqref{eq:LSsolution}.

\subsection{Iterative approach}

In this section, we present the iterative approach to design the ARMA coefficients. The idea consists of updating recursively the filter coefficients, while minimizing the original error~\eqref{eq:optimal}. We first reformulate the problem to make it amenable to our iterative approach and then use a variant of the Steigliz-McBride method~\cite{M} to implement an iterative algorithm that can be utilized for finding the ARMA graph filter coefficients.

\textbf{Problem reformulation.}The focus in the previous section was on solving \eqref{eq:prony}. This of course comes with a lack of optimality, since our aim is to solve \eqref{eq:optimal}. In the iterative approach, instead, we focus directly on minimizing \eqref{eq:optimal}.

To ease the notation, let us define
\[\beta _n = \sum\limits_{q = 0}^Q {{b_q}\lambda _n^q}\quad \text{and}\quad \alpha _n = \sum\limits_{p = 0}^P {{a_p}\lambda _n^p}, \]
and rewrite the original error~\eqref{eq:optimal} as
\begin{equation}\label{eq:oneerror}
e_n = {{{\hat h}_n} - \frac{{\beta _n}}{{\alpha _n}}}.
\end{equation}
Then, by defining ${\gamma _n = {1 \mathord{\left/
 {\vphantom {1 {\alpha _n}}} \right.
 \kern-\nulldelimiterspace} {\alpha _n}}}$, we have
\begin{equation}\label{eq:one1error}
e_n = {{{\hat h}_n}- \beta _n\gamma _n},
\end{equation}
which can be equivalently expressed as
\begin{equation}\label{eq:optimal2}
e_n =  {({{\hat h}_n}\alpha _n - \beta _n)\gamma _n}.
\end{equation}
Note that the expression \eqref{eq:optimal2} is linear in $\alpha _n $, $\beta _n$ and $\gamma _n$, if each of them is treated as a separate variable. To avoid inversion issues when $\alpha_n = 0$, we can consider $\gamma_n = 1/(\alpha_n + \rho)$ for some $\rho \approx 0$. Note that if $\gamma _n$ is fixed, $e_n $ becomes linear in the variables $\alpha _n$ and $\beta _n$. This will be our starting point to minimize $e_n$ recursively. In each iteration, having found a new set of solutions for $\alpha _n$, $\beta _n$ we can then find $a_p$ and $b_q$ as well as update $\gamma _n$.

To follow the convention of the previous sections, we write \eqref{eq:optimal2} in a more convenient vector form, by defining the vectors ${\boldsymbol{\alpha}} = [{\alpha _1}, \dots , {\alpha _N }]^{\rm T}$, ${\boldsymbol{\beta}}  = [{\beta _1}, \dots , {\beta _N }]^{\rm T}$, and ${\boldsymbol{\gamma}}  = [{\gamma _1}, \dots , {\gamma _N }]^{\rm T}$. Then, the error vector ${\bf{e}} = [{e_1}, \dots , {e_N }]^{\rm T}$ containing the original error for all graph frequencies can be written as
\begin{equation}\label{eq:merror}
{\bf{e}} =  {[{\hat{\bf h}} \circ {\boldsymbol{\alpha }} - {\boldsymbol{\beta }}] \circ {\boldsymbol{\gamma }}}.
\end{equation}

\textbf{Iterative algorithm} Let ${\boldsymbol{\alpha}}^{(i)}$ and ${\boldsymbol{\beta}}^{(i)}$ respectively denote the estimates of the vectors ${\boldsymbol{\alpha}}$ and ${\boldsymbol{\beta}}$, at the $i$-th iteration. We can then find the value of ${\boldsymbol{\gamma }}$ as an element-wise inversion of ${\boldsymbol{\alpha}}^{(i)}$, which we label as ${\boldsymbol{\gamma }}^{(i)}$,
\begin{equation}\label{eq:update}
{{\boldsymbol{\gamma }}^{(i)}} = {\left[ {\begin{array}{*{20}{c}}
{\frac{1}{{\alpha _1^{(i)} + \rho }}}&{\frac{1}{{\alpha _n^{(i)} + \rho }}}& \cdots &{\frac{1}{{\alpha _N^{(i)} + \rho }}}
\end{array}} \right]^T}.
\end{equation}
Using this value for ${\boldsymbol \gamma}$, we obtain the updated error
\begin{equation}\label{eq:iteration2}
{\bf{e}}^{(i+1)} = ({\hat{\bf h}}  \circ {{\boldsymbol{\alpha }}}) \circ {{{\boldsymbol{\gamma }}^{(i)}} - {{\boldsymbol{\beta }}} \circ {{\boldsymbol{\gamma }}^{(i)}}},
\end{equation}
 which is linear in the unknown variables ${{\boldsymbol{\alpha }}}$ and ${{\boldsymbol{\beta }}} $. Minimizing this error leads to the updated values ${{\boldsymbol{\alpha }}}^{(i+1)}$ and ${{\boldsymbol{\beta }}}^{(i+1)} $. This procedure is then repeated till a desirable solution is obtained.

To formalize this iteration, and express it as a direct function of the true filter coefficients ${\bf a}$ and ${\bf b}$, we can reformulate~\eqref{eq:iteration2} as
\begin{equation}\label{eq:newIT2}
{{\bf{e}}^{(i + 1)}} = {{{\bf{H}}^{(i)}}{{\bf{a}}} - {{\bf{B}}^{(i)}}{{\bf{b}}}},
\end{equation}
where ${{{\bf{H}}^{(i)}}=({\boldsymbol{\gamma}^{(i)}} {\bf 1}_{P+1}^{\rm T}) \circ {{\bf{\Psi }}_{P + 1}} \circ ({\hat{\bf h}}{\bf 1}_{P+1}^{\rm T})}$ and ${{{\bf{B}}^{(i)}}=({\boldsymbol{\gamma}^{(i)}} {\bf 1}_{Q+1}^{\rm T}) \circ {{\bf{\Psi }}_{Q + 1}}}$. The specific derivations that lead to~\eqref{eq:newIT2} can be found in Appendix B.

With this in place, the filter coefficients at the $(i+1)$-th iteration are found by solving
\begin{equation}\label{eq:iteration3}
\begin{array}{l}
\mathop {\min }\limits_{{{\bf{a}}},{{\bf{b}}}} {\left\| \begin{bmatrix}
{{{\bf{H}}^{(i)}}}, { -{{\bf{B}}^{(i)}}} \end{bmatrix} \begin{bmatrix}
{{\bf{a}}}\\
{{\bf{b}}}
\end{bmatrix} \right\|^2} \ \text{s.t.} \  a_0 = 1.
\end{array}
\end{equation}
The solutions ${\bf a}^{(i+1)}$ and ${\bf b}^{(i+1)}$ are again real-valued as shown in the following Proposition.
\begin{proposition}\label{th:4}
Under Property~\ref{prop:design}, the ARMA filter coefficients ${\bf a}$ and ${\bf b}$ obtained by solving~\eqref{eq:iteration3} are real-valued.
\end{proposition}
\begin{proof}
The proof is similar to the proof of Proposition~\ref{th:1}.
\end{proof}
For the above two design methods, the design cost of Prony's method is related to the LLS solution which requires $O((P+Q+1)^2N)$ operations, while for the iterative approach, the total design cost is $\tau $ times leading to a cost of $O(\tau(P+Q+1)^2N)$. Since the number of nodes $N$ is much smaller than the number of edges $E$, the design cost is smaller than the implementation cost. Algorithm 2 summarizes the iterative approach.

\begin{table}[!t]\renewcommand{\arraystretch}{1.2}
\centering
\begin{tabular}{rrl}
     \hline
     \hline
      \multicolumn{3}{l}{ $\bf Algorithm 2$: Iterative approach} \\
       \hline
       1&$\bf Input$: & ${{\bf{a}}^{(0)}}$, ${\hat{\bf h}}$, number of iterations ${\tau}$, threshold $\delta_c$\\
       2&$\bf Initialization$: & ${\boldsymbol{\gamma }^{(0)}}, {{\bf{H}}^{(0)}}, {{\bf{B}}^{(0)}} $, ${{{\hat{\bf g}}}^{(0)}}$, ${{\bf{e}} ^{(0)}}$\\
       3&$\bf Iteration $ : & while $i < {\tau}$ and ${\delta } < {\delta _{c}}$\\
       4& &  solve $\begin{array}{l}
\mathop {\min }\limits_{{{\bf{a}}},{{\bf{b}}}} {\left\| \begin{bmatrix}
{{{\bf{H}}^{(i)}}}, { -{{\bf{B}}^{(i)}}} \end{bmatrix} \begin{bmatrix}
{{\bf{a}}}\\
{{\bf{b}}}
\end{bmatrix} \right\|^2} \ \text{s.t.} \  a_0 = 1.
\end{array}$ \\
        5&&  return ${{\bf{a}}^{(i+1)}}$, ${{\bf{b}}^{(i+1)}} $ \\
        6&&  compute ${{{\hat{\bf g}}}^{(i+1)}}$, ${{\bf{e}} ^{(i+1)}}$, $\delta  = \| {{{\bf{e}} ^{(i + 1)}} - {{\bf{e}}^{(i)}}} \|$ \\
        7&&  update ${\boldsymbol{\gamma }^{(i+1)}}$\\
        8&&  $i=i+1$\\
        9&$\bf Output$: & ${{\bf{a}}^{(i+1)}}$, ${{\bf{b}}^{(i+1)}}$\\
       \hline
       \hline
   \end{tabular}
\end{table}

\textbf{ Remark 1.} We stop the iterations when $\delta$, representing the error difference between two successive iterations, is smaller than a given threshold $\delta_c$. However, depending on the specific combination of $P$ and $Q$, the method does not always converge fast enough or it does not converge at all. For those cases, we consider a maximum number of iterations $\tau$ and search for the minimum error over all iterations. We then assume that this iteration provides the solution to the problem. As we will see in the numerical section, for a fixed order $K$, the best performance for $P + Q \le K$ always leads to a significant improvement in approximation accuracy over the former methods.
However, for a fixed order $K$, some combinations of $P, Q$ yield instabilities around the cut-off frequency. The latter is especially present in Prony's method. Therefore, a search over different combinations of $P,Q$ is recommended.

\textbf{ Remark 2.} For ${{\boldsymbol{\gamma}^{(0)}} = {\bf{1}}}$, the LLS procedure \eqref{eq:LSsolution} can be seen as a special case of the iterative approach. With ${{\boldsymbol{\gamma}^{(0)}} = {\bf{1}}}$, the formulation of the iterative approach degenerates into the LLS solution, and the approximation error changes from the original error \eqref{eq:optimal} to the modified error \eqref{eq:prony}. However, since Prony's projection approach leads to better results that Prony's LS approach, we prefer the latter to initialize the iterative approach.

\begin{figure*}[!t]
\centering
\subfloat[]{\includegraphics[trim={1.5cm 0 2cm 0},clip,width=.48\textwidth]{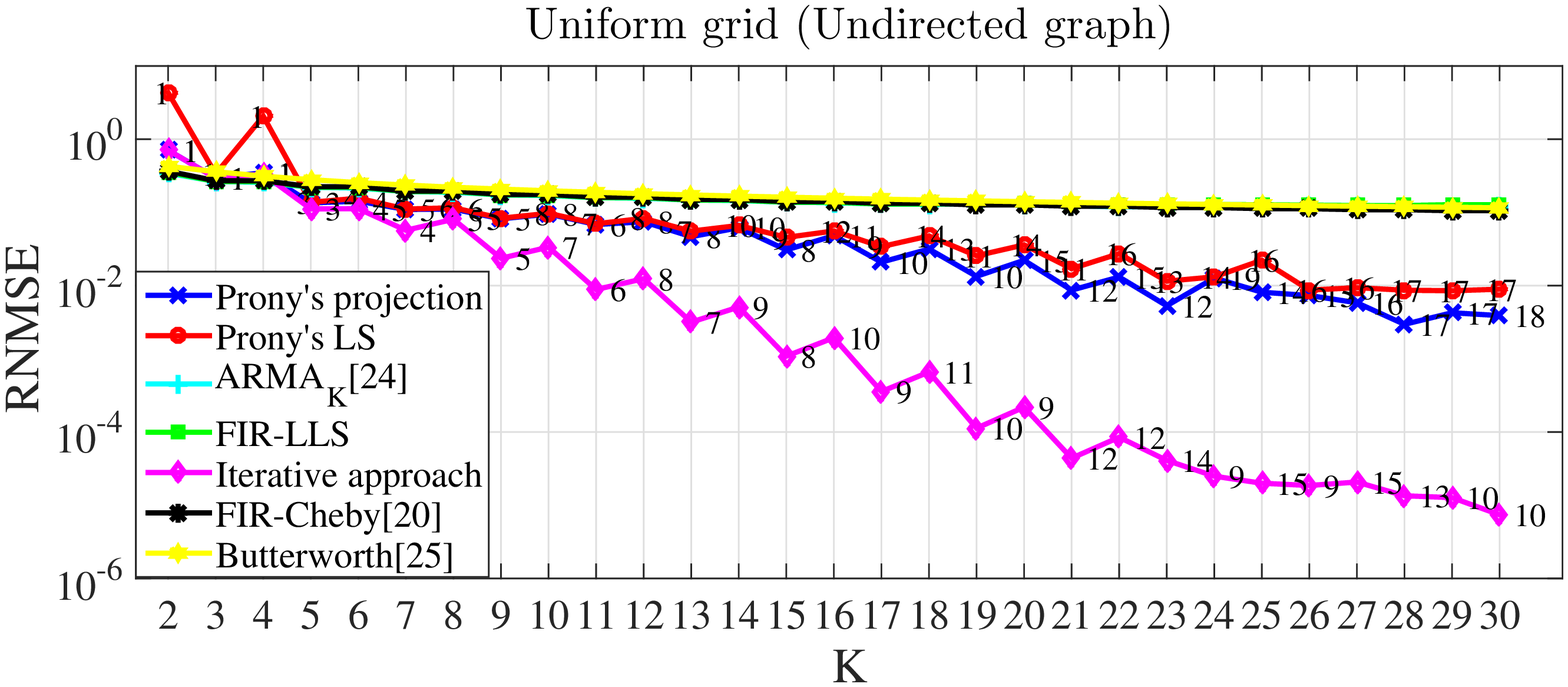}%
\label{subfig1}}
\subfloat[]{\includegraphics[trim={1.5cm 0 2cm 0},clip,width=.48\textwidth]{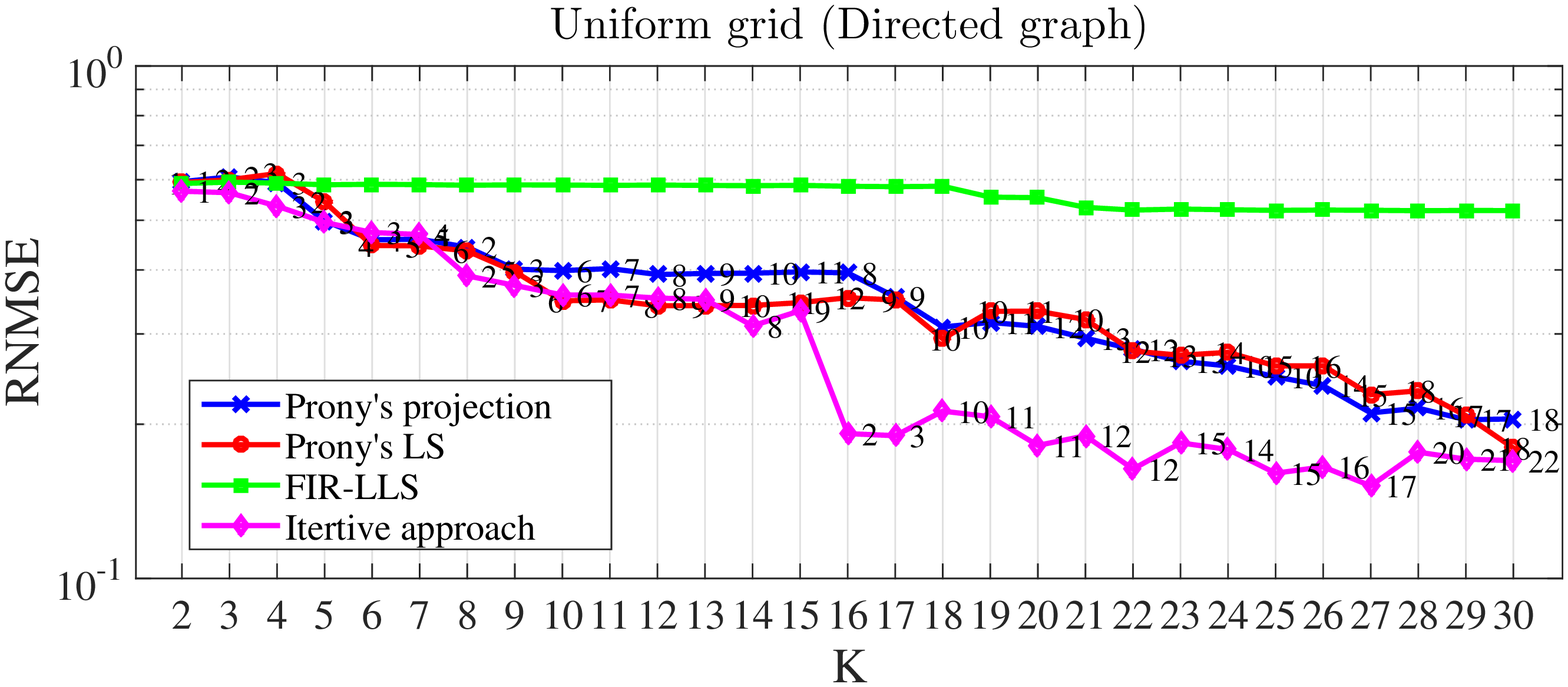}%
\label{subfig2}}\\
\subfloat[]{\includegraphics[trim={1.5cm 0 2cm 0},clip,width=.48\textwidth]{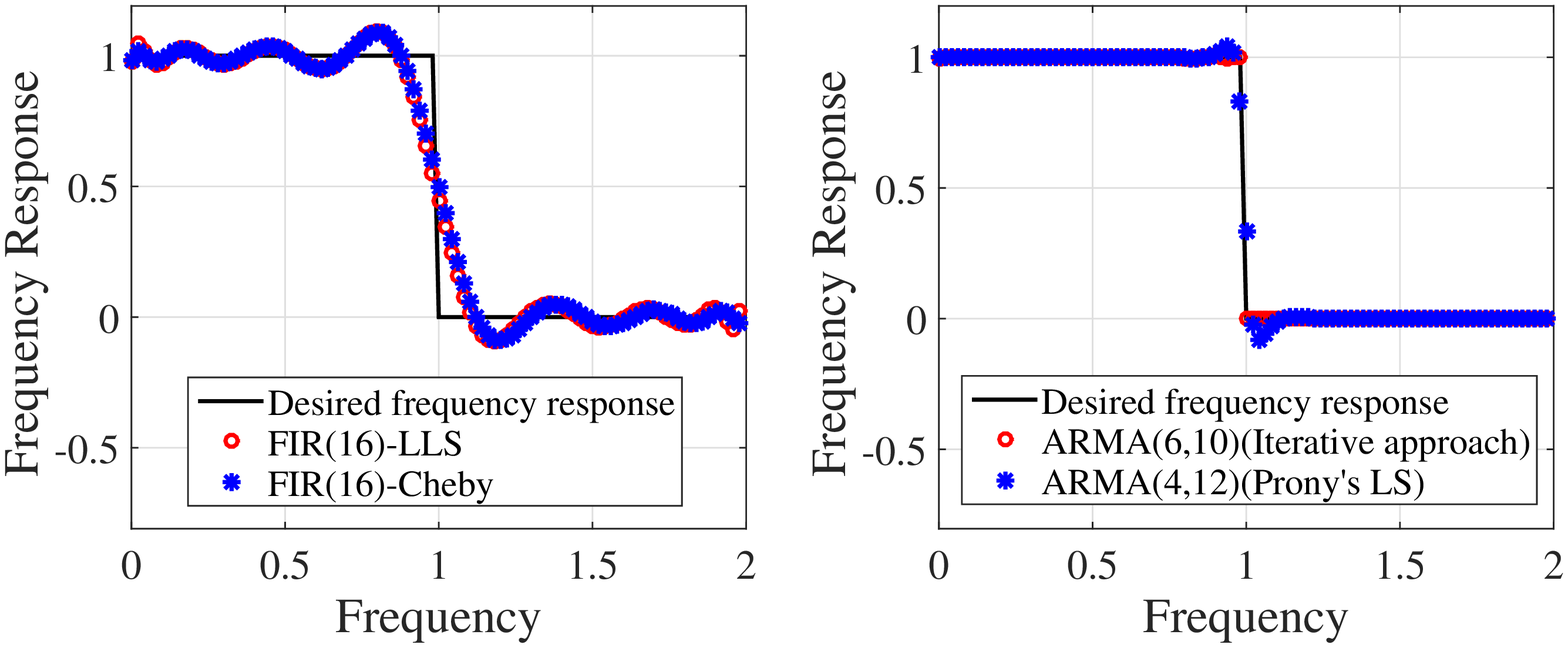}%
\label{subfig3}}
\subfloat[]{\includegraphics[trim={1.5cm 0 2cm 0},clip,width=.48\textwidth]{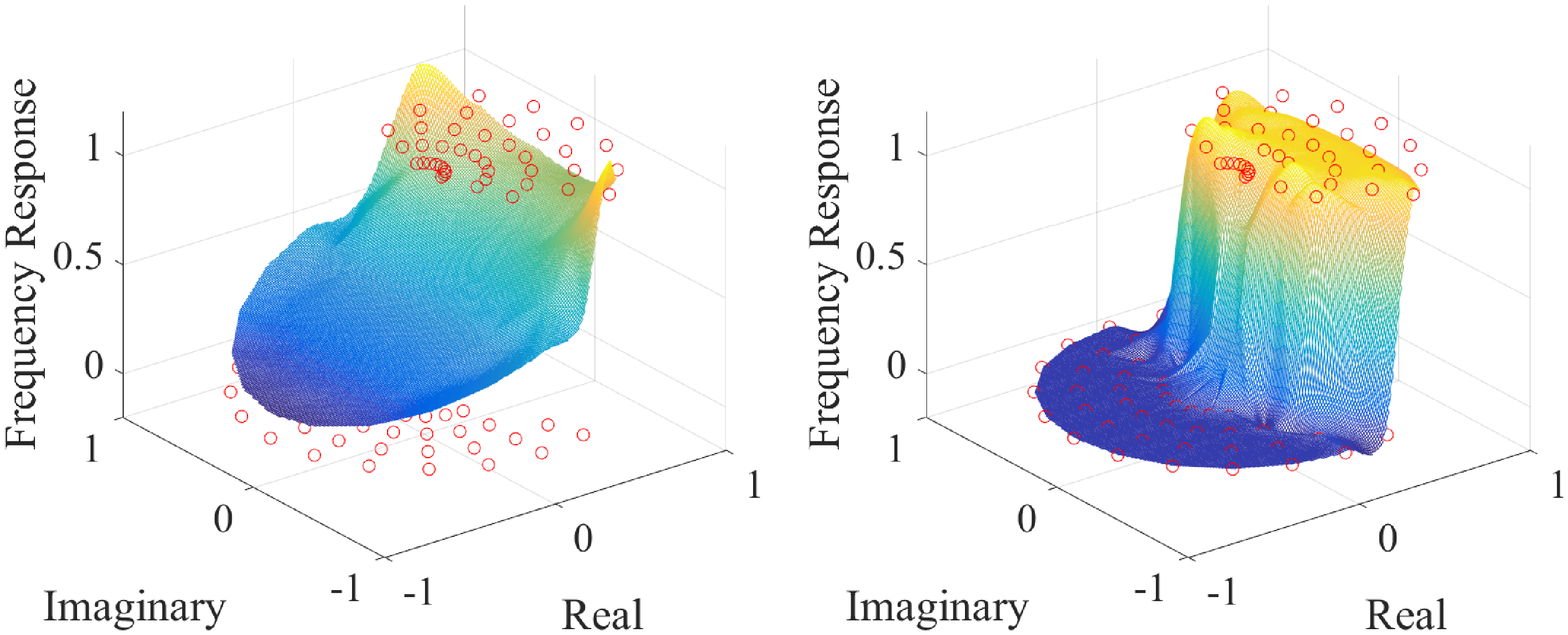}%
\label{subfig4}}\\
\subfloat[]{\includegraphics[trim={1.5cm 0 2cm 0},clip,width=.48\textwidth]{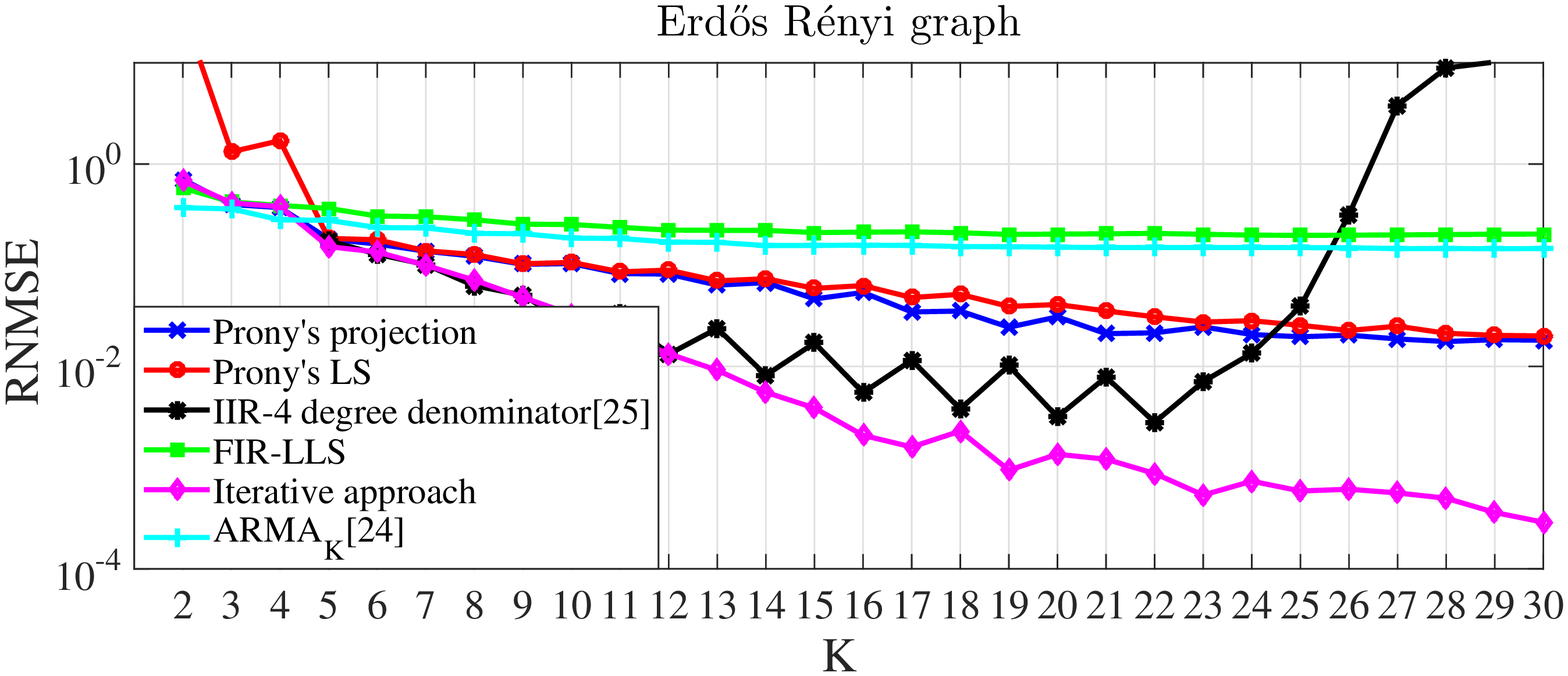}%
\label{subfig5}}
\subfloat[]{\includegraphics[trim={1.5cm 0 2cm 0},clip,width=.48\textwidth]{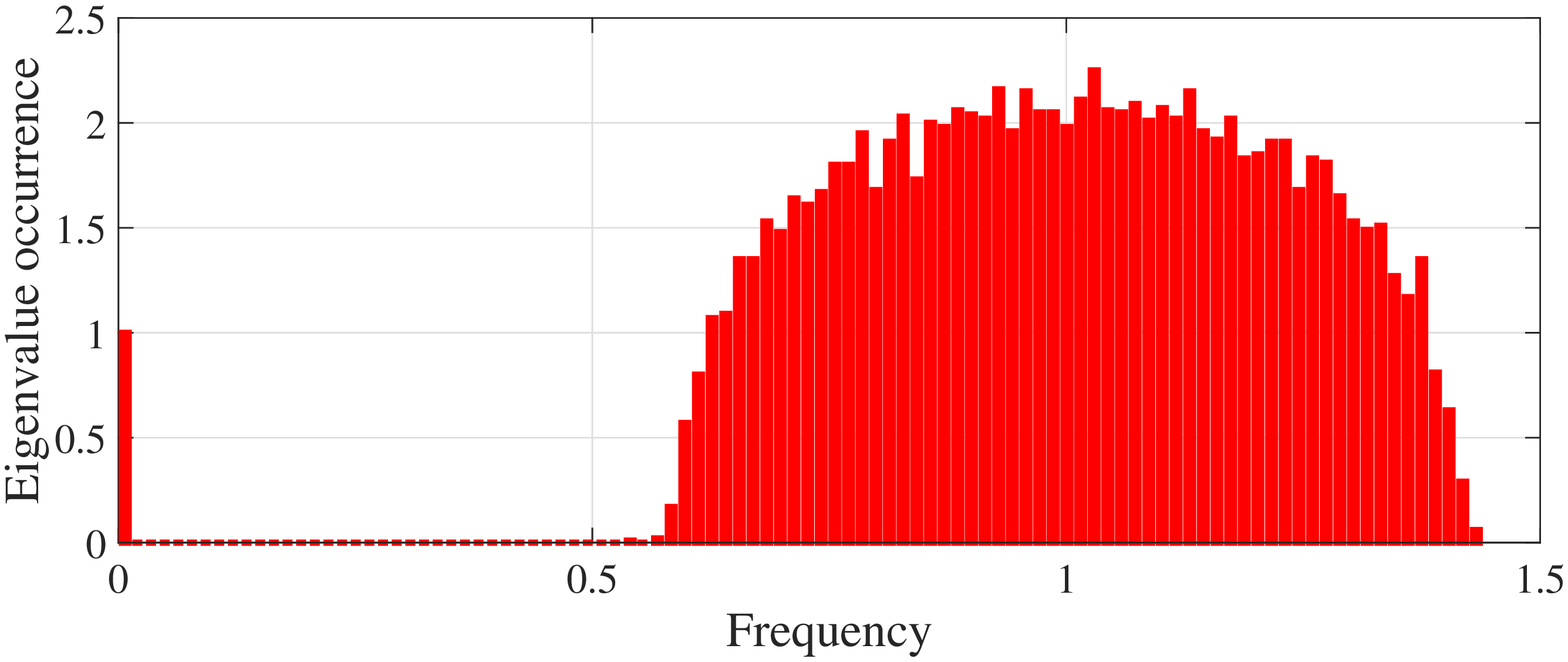}%
\label{subfig6}}
\caption{RNMSE of the proposed design methods for different orders $K$ (such that $P+Q = K$) in approximating an ideal low-pass frequency response. (a) Universal design by gridding the spectrum in $N = 100$ ($\bf{S} = \bf{L}_\text{n}$) points. For the ARMA filters, the order $Q$ is shown in the plot. (b) Universal design with $N = 100$ (${\bf{S} = \bf{A}_\text{n}}$) points. (c) Comparison of FIR and ARMA with same order $K=16$ for an undirected graph. The graph filters correspond to Fig. \ref{fig:simMP}(a). (d) Comparison of FIR and ARMA with same order $K=16$ for a directed graph. The FIR graph filter (left) and ARMA graph filter (right) correspond to the green and pink lines in Fig. \ref{fig:simMP}(b). The desired frequency response is shown in the plot as red points. (e) Results for the average of 100 Erd\H{o}s R\'{e}nyi graphs with $N = 100$ nodes and $p = 0.1$. (f) Eigenvalue occurrence of 100 Erd\H{o}s R\'{e}nyi graph realizations.}
\label{fig:simMP}
\end{figure*}

\section{Numerical data}
\label{Sec.Numerical data}
In this section, we present our numerical evaluation of the proposed methods and compare them with the FIR graph filters. The performance is tested with both synthetic and real data. Our tests with the Molene dataset\footnote{Access to the raw data is through the link: \url{https://donneespubliques.meteofrance.fr/donnees_libres/Hackathon/RADOMEH.tar.gz } } show that ARMA filters are more suitable than FIR filters for lossy data compression, where we can save up to 50\% of memory with very little error. Further, we apply ARMA filters in the context of prediction (as in~\cite{AJ}) and we show that ARMA graph filters outperform FIR graph filters, where with only 4 bits we achieve a reconstruction error of $10^{-3}$. Throughout our simulations we make use of the GSPBox \cite{NYPM}.
%

\subsection{Synthetic simulation results}

In this section we evaluate the performance of the proposed design algorithms in approximating a desired frequency response. The performance is assessed for two different settings, namely a \emph{universal} filter design (see Section~\ref{sec:universal}) and a filter design for an Erd\H{o}s  R\'{e}nyi (ER) graph. For both cases we consider $N = 100$ grid points / nodes \footnote{We remark that more grid points / nodes, i.e., $N = 300, 1000$, result in similar errors and trends as for $N = 100$.}. In both settings, the goal is to approximate the ideal low-pass frequency response introduced in Section~\ref{Sec.Pre} and illustrated in Fig.~\ref{fig:filg}.

\emph{Universal design:} For the universal design, we follow the approach discussed in Section~\ref{Sec.Pre}. For an undirected graph, we consider ${\bf S} = {\bf L}_{\text{n}}$ and sample the interval $[0,2]$ uniformly. For a directed graph, we consider ${\bf S} = {\bf A}_{\text{n}}$ and sample the complex unit disc uniformly in amplitude and phase. We assume $N=100$ grid points for both types of graphs.

\emph{Design for Erd\H{o}s  R\'{e}nyi graph:} For the undirected ER graph~\cite{PA}, we assume that a pair of nodes is connected with a probability $p = 0.1$ and the shift operator is again $\bf{S} = \bf{L}_\text{n}$. Due to the graph randomness we always average the results over 100 different realizations.

In the sequel, we analyze the design methods proposed in Section~\ref{sec.ARMAdes} and compare them to the related FIR filter design. If not mentioned otherwise, we design the FIR filter using the LLS approach of~\eqref{FIRpinv} (FIR-LLS, or simply FIR). The universal FIR design for undirected graphs sometimes also follows the Chebyshev design of~\cite{shuman2011distributed} (FIR-Cheby). We compare the ARMA$(P, Q)$ filter to a FIR$(K)$ graph filter where $P+Q \leq K$ is satisfied. We look for all combinations of $P$ and $Q$ that satisfy $P+Q \leq K$ and pick the combination leading to the best result. Since we want the overall order of the designed ARMA graph filter to be small, we only investigate the range $2 \leq K \le 30$. We measure the approximation accuracy with the root normalized mean square error (RNMSE) of the frequency response of the filter:
\begin{equation}\label{FIRerror}
\text{RNMSE} = \frac{{\| {{\hat{\bf h}} - {\hat{\bf g}}} \|}}{{\| {\hat{\bf h}}\|}}.
\end{equation}
Note that, for a directed graph with complex frequencies, since the filter response can be complex-valued, we only compute the approximation error for the amplitude (absolute value) of the filter response under the assumption that the desired frequency response is real.

\textbf{Performance analysis.}
In Fig. \ref{fig:simMP} we show the RNMSE for the Prony's inspired methods and the iterative approach. Specifically, the depicted RNMSE in Fig. \ref{fig:simMP} (a) (b) and (e) are related to the best combination $(P,Q)$ for each particular $K$ such that $P+Q = K$. The iterative approach is initialized with the solution of Prony's projection method ~\eqref{eq:LS5} and~\eqref{eq:LS6}, to show its potential in improving the RNMSE. Additionally, the FIR, ARMA$_K$~\cite{EA} and IIR~\cite{XM} performances are plotted as a benchmark.

Based on these results we can make the following observations:

i) We can notice that the FIR (FIR-LLS or FIR-Cheby) approximation errors for both universal designs (Fig. \ref{fig:simMP}(a), (b)) and the design for the ER graph (Fig. \ref{fig:simMP}(e)) is the highest, except when $K \le 5$. Further, the FIR approximation accuracy, even when designed for the specific set of ER graph frequencies, does not improve with the order $K$. We believe that this effect is due to the eigenvalue spread of the ER graph, since some of its eigenvalues are more closely spaced than in a uniform grid (see e.g., Fig. \ref{fig:simMP}(f)).

ii) Compared to Prony's method, the iterative approach has a larger design cost but improves the approximation for higher order $K$. Prony's method gives a comparable performance to the iterative approach only up to $K = 8$.
We see that Prony's LS approach is not suitable for the ER graph when $K \le 5$, while for a universal design approach its performance is close to that of Prony's projection method. This highlights that the LS approach should be avoided in graphs that have closely spaced eigenvalues. On the other hand, this issue is overcome by Prony's projection method which gives a small RNMSE also for values $K \le 5$.

iii) As an example, we take the order $K=16$ to show the difference in performance between FIR graph filters and ARMA graph filters in Fig. \ref{fig:simMP}(c), (d). It is remarkable to highlight that the iterative approach outperforms the FIR by several orders, where the latter has a comparable performance only for $K \le 3$. Such a finding shows that the ARMA graph filters are more suitable for applications demanding higher approximation accuracies.

iv) We observe a smaller RNMSE for undirected graphs compared to directed graphs. This is because we can do a fitting on the real line instead of in the complex plane. In contrast to undirected graphs, notice that for directed graphs, as shown in Fig. \ref{fig:simMP} (b), all ARMA graph filter design approaches yield a similar performance.

v) As highlighted in Fig. \ref{fig:simMP} (a), an important role is played by the MA order $Q$ (which is generally larger than $P$). We observe that a higher $Q$ improves the stability of the ARMA filters, specifically for Prony's projection method and the iterative approach where the numerator coefficients are found by minimizing the true error.

vi) If the frequencies are different, the Vandermonde matrix ${\bf \Psi}$ is theoretically full rank (invertible) but generally ill-conditioned. Although this issue is encountered for both FIR and ARMA graph filters, ARMA filters improve the conditioning of the matrix because the filter orders $P$ and $Q$ can be selected much lower than the FIR filter order $K$. Hence, the solution of our design methods has uniqueness, but there might be a conditioning problem when the orders are increased.

vii) For the universal design (Fig. \ref{fig:simMP}(a)) and ER graph (Fig. \ref{fig:simMP}(e)), we also compare our approach with the methods in~\cite{EA,XM}. The ARMA$_K$ graph filter~\cite{EA} has the same order for the nominator and denominator, therefore, we adopt the same value $K$ as order for both the nominator and denominator. Note that this leads to a total order that is twice the order of our ARMA$(P,Q)$ (recall that $K \leq P+Q$). For the universal design, we further compare our approach with the universal Butterworth filter~\cite{XM}. The IIR graph filter~\cite{XM} is then tested on the ER graph. We follow the scenario of~\cite{XM} and use a denominator of degree $4$, leading to a nominator of degree $(K-4)$. The results show that for low orders ($K < 12$), the IIR graph filter~\cite{XM} has a similar performance to our iterative approach. However, with an increasing order $K > 12$, our design method offers a better approximation accuracy.

\begin{figure}[!t]
\centering
\includegraphics[trim={1.5cm 0 2cm 0},clip,width=.48\textwidth]{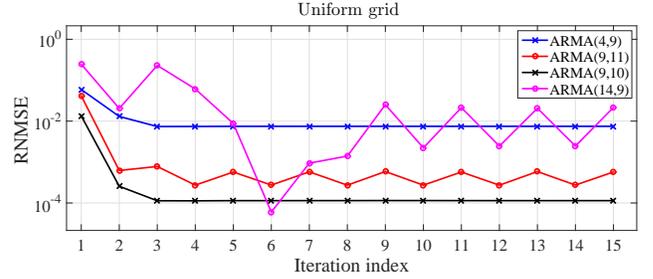}%
\caption{RNMSE of the iterative approach on the universal design with $N=100$ points. Performance evaluation for different ARMA filters which are a few particular cases illustrating monotonic convergence, non-monotonic convergence, and no convergence. }
\label{fig:simMP2}
\end{figure}

\textbf{Iterative approach.} We now analyze in more detail the iterative approach, to highlight its benefits in improving the ARMA filter accuracy compared to Prony's projection approach. We consider two cases with monotonic convergence, namely, an ARMA$(9, 10)$ (characterized by an RNMSE of order 10$^{-2}$ in Prony's projection method, Fig. \ref{fig:simMP} (a)) and an ARMA$(4, 9)$ (characterized by an RNMSE of order 10$^{-1}$ in Prony's projection method, Fig. \ref{fig:simMP} (a)) which are considered due to their low orders. For both cases we initialize the iterations with the solution of Prony's projection method. Note that ARMA$(9, 10)$ is the best combination $P, Q$ of order $K=19$, while ARMA$(4, 9)$ is not the best combination for order $K=13$.
We also consider two filters, the ARMA$(9, 11)$ and ARMA$(14, 9)$ to illustrate that even without monotonic convergence, the approximation accuracies can be improved with our iterative approach.

In Fig. \ref{fig:simMP2} we show the approximation error as a function of the iteration index and we can immediately notice that for those filters with monotonic convergence, the approximation errors reduce in a few iterations. More specifically, for the ARMA$(9, 10)$ the iterative approach reduces the error from $10^{-2}$ to $10^{-4}$. It is also worth noticing that using the iterative approach, the ARMA$(9, 10)$ outperforms also the ARMA$(11, 17)$, which is the best filter that can be designed with Prony's projection method (within the considered range). Similarly, the iterative approach improves the approximation accuracy for the low order filter ARMA$(4, 9)$. Indeed, its performance is now comparable with all other ARMAs and FIRs with much greater orders.
As we mentioned in the previous section, for the non-converging filters, we pick the best approximation result during the iterative procedure, e.g., the performance in the $6$-th iteration of the ARMA$(14, 9)$ filter, which is better than the performance of the ARMA$(9, 10)$ filter.

We remark that the above results concern the approximation accuracy of the filter irrespective of their implementation costs. In the sequel, we address some implementation aspects.

\begin{figure}[!t]
\centering
\includegraphics[trim={1.5cm 0 2cm 0},clip,width=.48\textwidth]{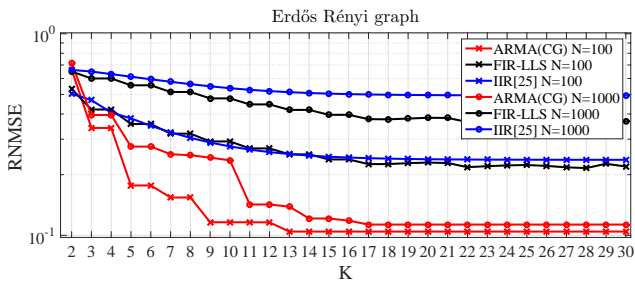}%
\caption{RNMSE of the ARMA graph filter implementation on an Erd\H{o}s R\'{e}nyi graph with $N=100$ and $N=1000$. Performance evaluation for the trade-off between computational cost and approximation accuracy. For CG, the complexity of the ARMA implementation is limited by ${PT + Q \le K}$.}
\label{fig:simCG}
\end{figure}

\textbf{CG implementation performance.}
We now aim at analyzing the ARMA implementation performance using the CG approach w.r.t.~its implementation cost. We implement the universally designed ARMA filter using CG on the ER graph with link probability $p = 0.1$, and consider two different sizes: $N=100$ and $N=1000$. We again use the universally designed FIR and IIR graph filters as benchmarks.
The ARMA filter coefficients are designed universally using the iterative approach with 100 grid points, whereas the FIR filter is designed using LLS also with 100 grid points and the IIR filter following the Butterworth approach~\cite{XM}.
The filter is applied to a white input and the desired frequency response (low pass filter) is compared to the division of the filter output and the input in the frequency domain. In Fig.~\ref{fig:simCG}, we show the performance of the ARMA filter (Algorithm 1) when the CG is halted after $T$ iterations such that $PT + Q \le K$ holds, i.e., the ARMA filter has a smaller or the same implementation cost compared to the FIR filter. For the CG, we set $\varepsilon  = {10^{ - 3}}$.
The IIR filter has the same order $K$ as the FIR filter and is given a maximum number of iterations of $T=30$. The results show that the ARMA filter has a lower approximation error than other alternatives with a similar or smaller complexity.
Since we here compare the filters for a similar implementation complexity, the RNMSE gap is smaller compared to the previous scenario in Fig. \ref{fig:simMP} (a) and (e).
To highlight the benefits of the universal design approach, we consider the ER graph with two different sizes. In Fig.~5, we notice that when increasing the filter order ($K>16$), the performance of the ARMA graph filter for different size graphs becomes similar. Even for the case with $N=1000$, the universal design based on 100 grid points is a wise choice and yields a good performance.

Although the aforementioned results are obtained using synthetic data, they highlight the potential of ARMA filters to improve the performance w.r.t. FIR graph filters.
The above results can be useful in practice for spectral clustering; building graph filter banks, or designing graph wavelets, where we propose the use of ARMA filters instead of FIR filters.

As we will see next, this improvement in performance of ARMA filters is also present in real data applications.

\subsection{Graph signal interpolation.}

\begin{figure}
  \centering
  \includegraphics[trim={1.5cm 0 2cm 0},clip,width=.48\textwidth]{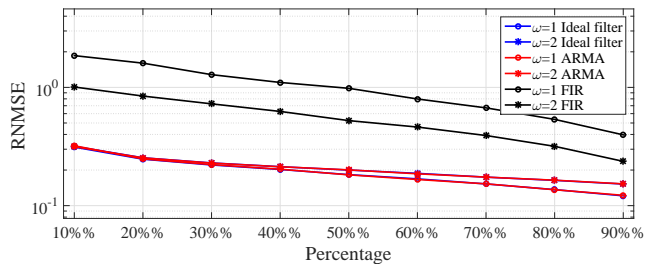}
  \caption{RNMSE of the ARMA graph filter for interpolation of the Molene data set, where $ \omega =1, 2$. As two comparisons, ideal graph filter and FIR filter with order $K=20$ are shown with the same values of $\omega $.}
  \label{fig:Inter}
\end{figure}

We now illustrate the performance of ARMA graph filters in interpolating the missing values in the Molene weather data set. The data set contains hourly observations of temperature measurements collected in January 2014 in the region of Brest (France). The undirected graph containing the 32 cities (nodes) is built according to \cite{SSG}, which accounts for the smoothness of the data w.r.t. the graph structure. We consider that a portion of the graph signal is missing, and by exploiting the smoothness prior we aim to reconstruct the overall graph signal from noisy measurements.

{\textbf{ Experimental set up.}} Given ${\bf x}^\prime$ the observed signal and ${\bf x}$ the original graph signal, this interpolation problem is formulated as \cite{SAA} \cite{mao2014image}:

\begin{equation}\label{eq:inter1}
\mathop {\min }\limits_{\bf{x}} \left\| {{\bf{T}}({\bf{x}} - {{\bf{x}}^\prime })} \right\|_2^2 + \omega {{\bf{x}}^T}{{\bf{L}}_{{n}}}{\bf{x}}
\end{equation}
where ${\bf{T}}$ is a diagonal matrix with ${{T_{ii}} = 1}$ if ${x_i}$ is known and ${{T_{ii}} = 0}$ otherwise; ${\omega}$ is the weight for the prior.
The optimal solution of \eqref{eq:inter1} is
\begin{equation}\label{eq:inter2}
{\bf{\tilde x}} = {({\bf{T}} + \omega {{\bf{L}}_{{n}}})^{ - 1}}{{\bf{x}}^\prime },
\end{equation}
which by considering ${\bf{P}} = {\bf{T}} + \omega {{\bf{L}}_{{n}}}$ is solved through the ARMA graph filter \eqref{L4}.
We consider the CG to implement \eqref{eq:inter2} where $\varepsilon$ is set to $10^{-2}$ and the maximum number of iterations $T$ to $20$. As a comparison, for the FIR graph filter, the coefficients are found as the solution of
\begin{equation}\label{eq:interFIR}
\mathop {\min }\limits_{{g_k}} \left\| {{{({\bf{T}} + \omega {{\bf{L}}_n})}^{ - 1}} - \sum\nolimits_{k = 0}^K {{g_k}{\bf{L}}_n^k} } \right\|_F^2
\end{equation}
where the ${g_k}$ values represent the FIR coefficients.

\textbf{ Results.} In Fig. \ref{fig:Inter} we show the RNMSE between the reconstructed signal ${\bf{\tilde x}}$ and the original one ${\bf x}$ as a function of the portion of missing data. Additionally, to construct the observed signal ${\bf x}^\prime$, we add a zero-mean Gaussian noise with variance ${\sigma ^2}=10^{-2}$ to the original signal ${\bf x}$ and randomly wipe off signals up to the specific percentage. The performance is average over all 744 observations. We plot the numerical RNMSE for different percentages and two $\omega$ values. These results show that the RNMSE reduces for the ARMA graph filter when the percentage of the known values increase.
As a comparison, we notice that the ARMA graph filter offers a similar performance to the ideal graph filter. The FIR graph filter ($K=20$) yields a worse result in this case.

\subsection{Data compression with graph filters}

\begin{figure}[!t]
\centering
\subfloat[]{\includegraphics[trim={1.5cm 0 2cm 0},clip,width=.48\textwidth]{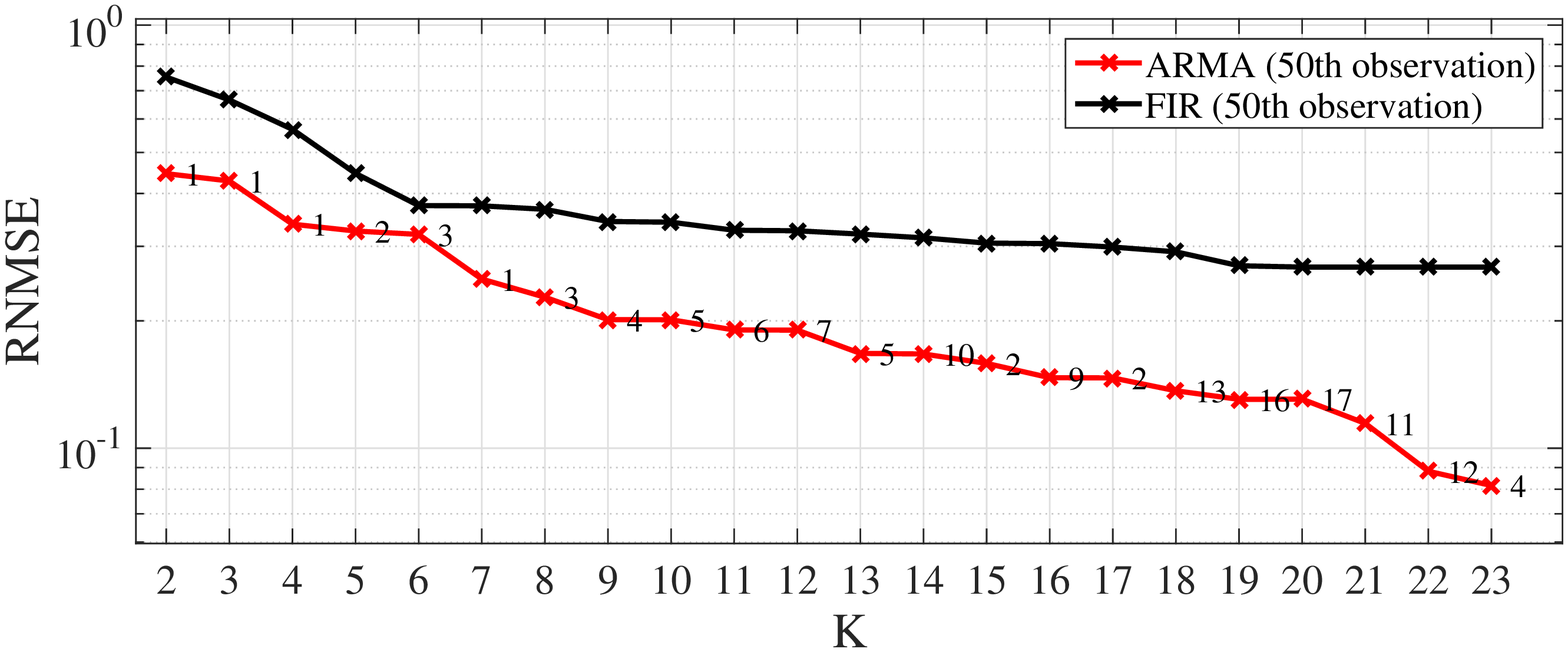}%
\label{subfig1}}\\
\subfloat[]{\includegraphics[trim={1.5cm 0 2cm 0},clip,width=.48\textwidth]{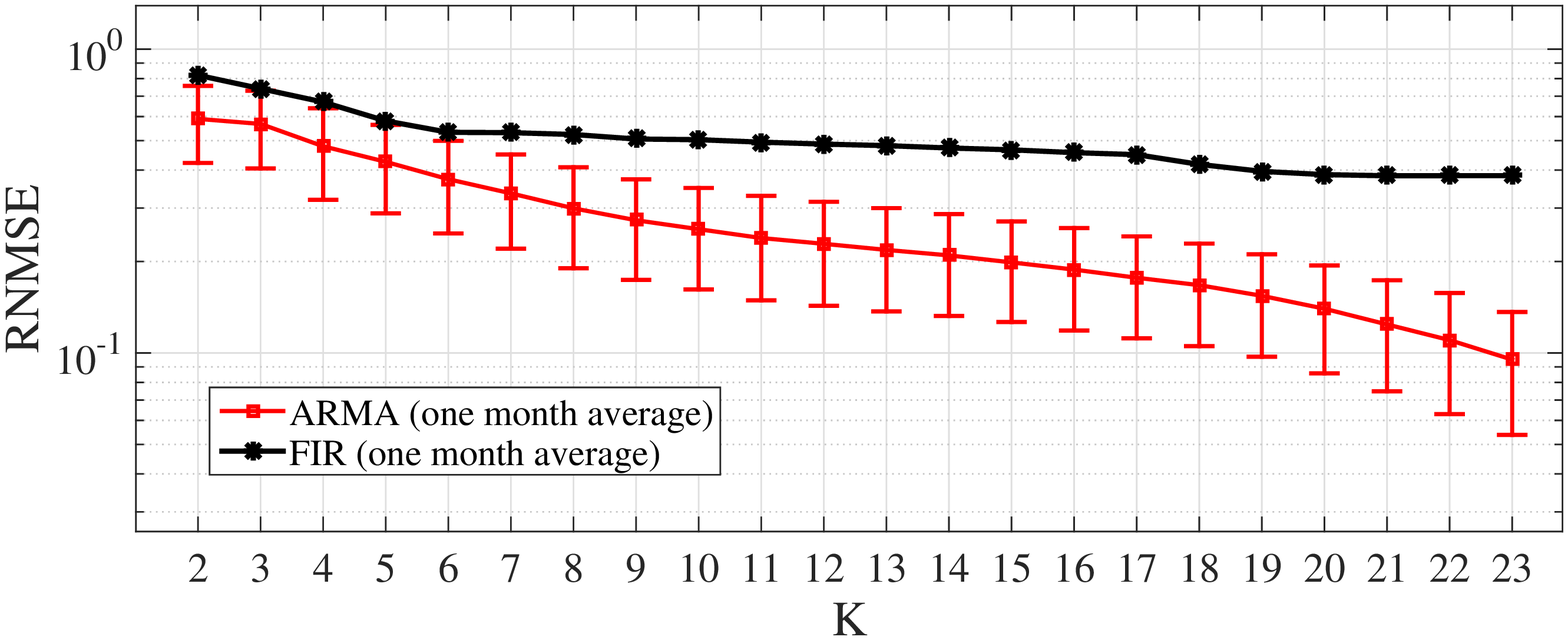}%
\label{subfig2}}
\caption{RNMSE between the data spectrum and the filter frequency response as a function of filter order $K$. (a) Illustration of the RNMSE of the ARMA graph filter and the same order FIR filter for the 50th observation. The order $Q$ is shown in the plot and $P + Q = K$. (b) Average RNMSE over all 744 temperature realizations (one month) for different filter orders. For the ARMA filter, each error bar shows the standard deviation of the approximation error for order $K$. }
\label{fig:MD}
\end{figure}

Our goal, in this subsection, is to show that ARMA filters of low orders can be used to represent the data and perform compression.

{\textbf{ Experimental set up.}} We consider fitting a small order ARMA graph filter to each data realization and then store the filter coefficients instead of the actual data. We now create the graph as a directed 6-nearest neighbor connection. In the directed graph, each vertex is connected to its six closest nodes by means of directed edges \cite{ASJ2}. The weight of the edge between $v_m$ and $v_n$ is given as
\begin{equation}\label{eq:createA}
[{\bf A}]_{n,m} = \frac{{{e^{ - d_{n,m}^2}}}}{{\sqrt {\sum\nolimits_{k \in {\mathcal N}_n} {{e^{ - d_{n,k}^2}}} \sum\nolimits_{l \in {\mathcal N}_m} {{e^{ - d_{m,l}^2}}} } }}
\end{equation}
where $d_{n,m}$ represents the geometric distance between nodes $v_n$ and $v_m$ and ${\mathcal N}_n$, ${\mathcal N}_m$ represent the sets of neighbors of node $v_n$ and $v_m$. Note that the resulting matrix ${\bf{A}}$ is normal, i.e., $\left\| {\bf{A}} \right\| = 1$. For every data realization $\bf{x}$, we take the GFT to have $\hat{\bf{x}}$ and fit it to an ARMA$(P, Q)$ graph filter. The filter coefficients are derived using the iterative approach with the initial condition given by Prony's projection method. We measure the compression performance as the RNMSE between the compressed signal and the real one $\bf{x}$. As a benchmark, we again consider the FIR$(K)$ with $K=P+Q$.

\textbf{ Results.} In Fig. \ref{fig:MD}(a), we show the RNMSE as a function of $K$ for the 50-th observation. We observe that the ARMA filter achieves a smaller RNMSE than the FIR filter even for small orders $K$. As expected, when $K$ approaches $N$, we have a smaller error but we also see that the gap in performance between the ARMA and FIR filters increases. This result goes in line with what we obtained in the previous section for synthetic data.

To further quantify the above observations, Fig. \ref{fig:MD}(b) depicts the average performance over all observations. We still notice that the ARMA graph filters achieve a smaller RNMSE than FIR graph filters, and that the RNMSE decreases for higher values of $K$. With the above approach, a compression ratio of $25{\rm{\% }}$ ($K=23$) is achieved with an RNMSE of $10^{-1}$ can be tolerated. Note that next to signal compression, the ARMA model can also be used to reconstruct the graph power spectrum of stationarity graph signals from a subset of the nodes \cite{chepuri2017graph}.

\textbf{ Remark 3.} To achieve further compression one can exploit also the stationarity of the signal over time. Thus, instead of fitting a graph filter to each individual observation, one approach may consider fitting a joint graph-temporal filter \cite{isufi2016separable}, \cite{isufi20162} to the time-varying data.

\subsection{Linear prediction with ARMA filters}

Inspired by \cite{AJ}, we also test linear prediction (LP) on graphs using ARMA graph filters. We consider the Molene data set and again compare the ARMA graph filters with the FIR graph filters~\cite{AJ}. The considered problem contains two parts, namely the forward (prediction) part and the backward (synthesis) part. In the forward filtering, the residual between the graph signal and the filter frequency response is calculated and quantized. Next, the backward filter considers building an approximation of the graph signal from the quantized residual. For the ARMA filters we use a variant of the iterative approach to find the filter coefficients, while for the FIR filter we follow \cite{AJ}. For the graph shift operator ${\bf{S}}$, we consider both the directed graph created by ~\eqref{eq:createA} and the undirected graph \cite{SSG}.

{\textbf{ Experimental set up.}} For the ARMA filter, given the graph signal ${\bf{x}}$, the residual ${\bf{r}}$ related to signal prediction is given by
\begin{equation}\label{eq:pred1}
{\bf{r}} = {\bf x} - g({\bf S}) {\bf x} = {\bf{x}} - {(\sum\limits_{p = 0}^P {{a_p}{{\bf{S}}^p}} )^{ - 1}}(\sum\limits_{q = 0}^Q {{b_q}{{\bf{S}}^q}} ){\bf{x}}.
\end{equation}
Notice that next to the constraint $a_0=1$ we had before, it is important to set $b_0=0$ in order to avoid the trivial solution.
Similar to Prony's method, we can derive also a modified residual as
\begin{equation}\label{eq:pred2}
{\bf r}' = (\sum\limits_{p = 0}^P {{a_p}{{\bf{S}}^p}} ) {\bf x} - (\sum\limits_{q = 0}^Q {{b_q}{{\bf{S}}^q}} ) {\bf x}. \end{equation}
To relate this prediction problem to filter design, we can look at the residual and modified residual in the frequency domain, leading to
\begin{equation}\label{eq:pred3}
\hat{\bf{r}} = \hat{\bf{x}} \circ ( {\bf 1}_N - \text{diag}( {\bf \Psi}_{P+1} {\bf a} )^{-1} {\bf \Psi}_{Q+1} {\bf b} ), \end{equation}
and
\begin{equation}\label{eq:pred4}
 \hat{\bf r}' = \hat{\bf{x}} \circ  [ {\bf 1}_N \circ ( {\bf \Psi}_{P+1} {\bf a} ) - {\bf \Psi}_{Q+1} {\bf b} ].
\end{equation}
Hence, up to the element-wise multiplication with $\hat{\bf{x}}$, this residual $\hat{\bf{r}} $ and modified residual $ \hat{\bf r}' $ look like the error ${\bf e}$ in~\eqref{eq:original_error} and modified error in ${\bf e}'$ in~\eqref{eq:other}, respectively, with $\hat{\bf h}$ replaced by the all-one vector ${\bf 1}_{N}$. As a result, all previous design methods can still be used. They only need to be adapted with an appropriate weighting (coming from $\hat{\bf{x}}$) and with the constraint $b_0=0$.

Once the filter coefficients that (approximately) minimize the residual ${\bf r}$ are found, this residual is quantized with $B$ bits (resulting in ${\bf{r}}_\text{q}$) and forwarded. Then, by applying the backward filter ${\bf{ H}} = {({\bf{I}} - g({\bf{S}}))^{ - 1}}$ to the residual, the approximated signal $\widetilde {\bf{x}} = {\bf{{ H}r_\text{q}}}$ is constructed at the receiving side.

We consider ARMA graph filters for $K \le 10$ ($K=P+Q$) and for every order $K$, the residual ${\bf{r}}$ is quantized with different numbers of bits. From the $B$ bits, we spend one bit on the sign, $ b = \left\lceil {{{\log }_2}(\max ( [{\bf{r}}]_i ))} \right\rceil $ bits on the integer part, and the rest of the $(B-b-1)$ bits on the decimal fraction.

\textbf{ Results.} We quantify the performance in terms of RNMSE between the predicted signal $\tilde{{\bf{x}}}$ and the original one ${\bf{x}}$.

\begin{figure}[!t]
\centering
\subfloat[]{\includegraphics[trim={1.5cm 0 2cm 0},clip,width=.48\textwidth]{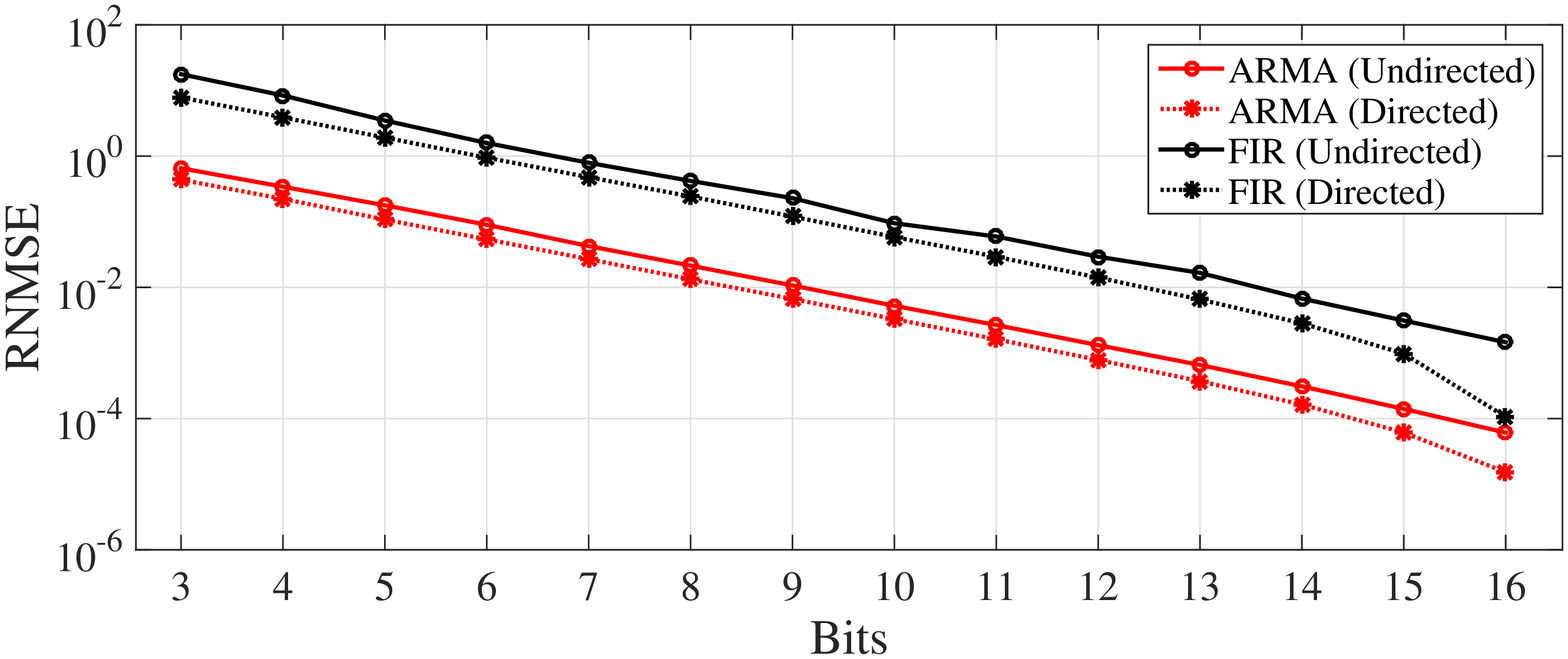}%
\label{subfig1}}\\
\subfloat[]{\includegraphics[trim={1.5cm 0 2cm 0},clip,width=.48\textwidth]{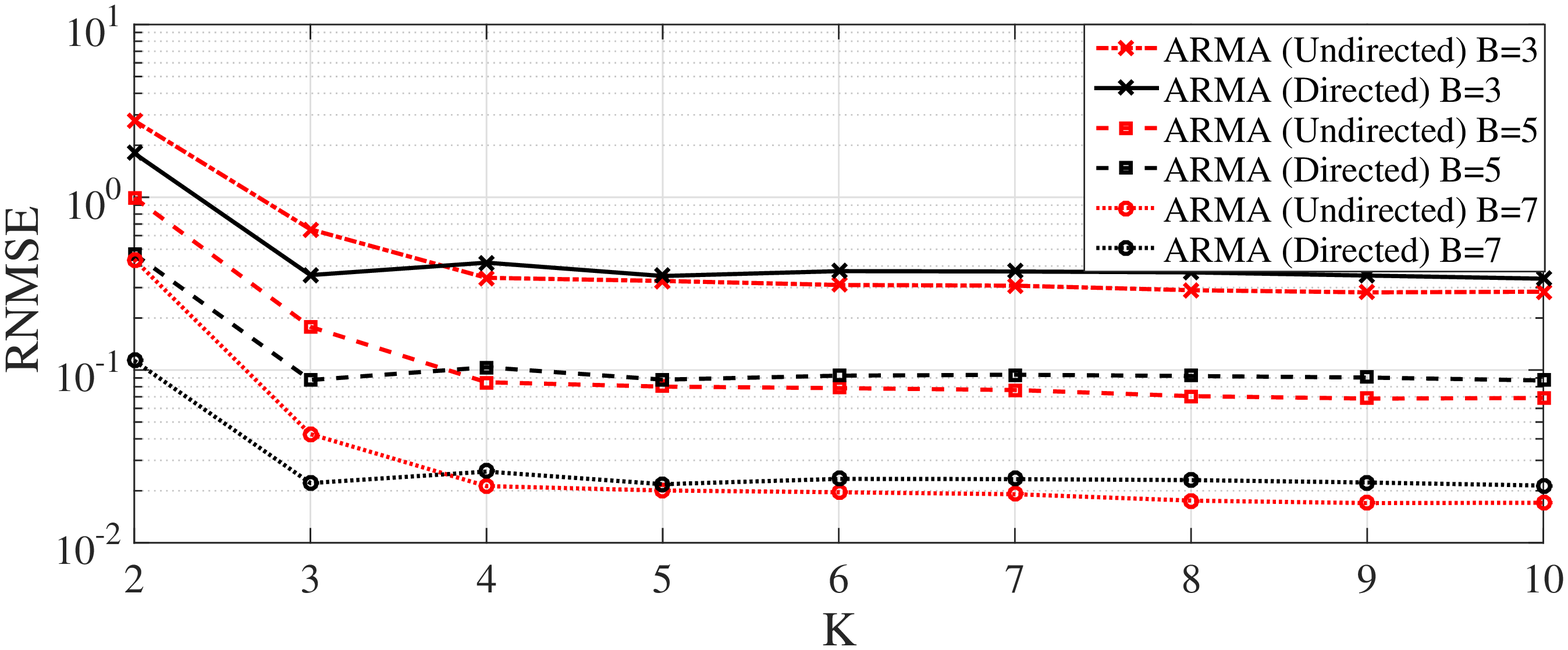}%
\label{subfig1}}
\caption{Average RNMSE of linear prediction on the Molene temperature data set. (a) Average RNMSE of the approximated signal as a function of the number of bits ($B$) for filter order $K=3$. (b) Average RNMSE of the estimated signal for different order ARMA filters evaluated for $B=3, 5, 7$.}
\label{fig:PD}
\end{figure}

The average approximation error over all 744 realizations is shown in Fig. \ref{fig:PD}(a) as a function of the number of bits $(B)$ used in the quantization for $K=3$. We can notice that in a direct comparison with the FIR filters the approximation error of the ARMA graph filters is more than one order of magnitude lower. For both filters, as expected, more quantization bits $B$ lead to a better approximation accuracy. Such findings suggest once again that ARMA filters are more suitable than FIR filters for applications demanding higher approximation accuracies.

To better highlight the performance of the ARMA filters, in Fig.\ref{fig:PD}(b) we show the RNMSE as a function of the filter order $K$ for different values of $B$. These results show that the approximation error for $K>4$ remains constant, similar to what was observed for FIR filters in \cite{AJ}. This observation suggests that small order filters are preferred for this application. Note that the performance for directed and undirected graphs is almost the same. The directed graph gives the best performance with $K=3$ while for $K>3$ the undirected graph gives a lower error. To conclude, we can say that using an ARMA graph filter with $K=4$ and $B=7$ (instead of 16 bits) we can reconstruct the data with an error of order ${10^{ - 2}}$ and save $62.5\% $ in transmission costs.

\section{Conclusions}
\label{Sec.Concl}
In this work, we have presented ARMA graph filters as well as different methods to perform the filter design on both directed and undirected graphs. The first two filter design approaches are inspired by Prony's method which focus on minimizing some modified errors. The third one iteratively minimizes the original error of the design problem. The iterative approach can be initialized with the solution from one of the previous methods, which suggests that its performance can be improved by the iterative approach. Our theoretical findings are surrogated by numerical results on both synthetic and real data. In a direct comparison with the FIR graph filters, ARMA filters have shown to be more suitable for filter approximation, data interpolation, data compression and linear prediction on graphs.

\appendices

\section{}
Since we assume that the shift operator ${\bf S}$ is real-valued and diagonalizable, the graph frequencies $\lambda _n$ (eigenvalues) can be grouped into three sets: $1 \le n \le M$, $M+1 \le n \le 2M$ and $2M+1 \le n \le N$. The first and second groups are complex conjugate pairs while the last group consists of the real-valued frequencies. Note that this classification only changes the order of the frequencies, and has no influence on the results of the filter coefficients ${g}_k$.

Thus, we can split the Vandermonde matrix ${{\bf{\Psi }}_{K + 1}}$ and write~\eqref{eq:FIRproblem} as
\[\begin{array}{l}
\mathop {\min }\limits_{\bf{g}} \| \hat {\bf{h}} - {{\bf{\Psi }}_{K + 1}}{\bf{g}}\| {^2}\\
 = \mathop {\min }\limits_{\bf{g}} {\| {[\hat {\bf{h}}_1^H,\hat {\bf{h}}_2^H,\hat {\bf{h}}_3^T]^T - [{\bf{\Psi }}_{1}^H,{\bf{\Psi }}_{2}^H,{\bf{\Psi }}_{3}^T]^T{\bf{g}}} \|^2}
\end{array}\]
where the three blocks of matrices and vectors belong to the three different groups.

With $1 \le n \le M$, we use the $n$th and $(M+n)$th frequencies to represent a conjugate pair for the first and second groups of frequencies. Since ${\lambda _n} = \lambda _{M+n}^*$, the corresponding elements inside the Vandermonde matrix satisfy
${[{{\bf{\Psi }}_{1}}]_{n,k}} = [{{\bf{\Psi }}_{2}}]_{M+n,k}^*$, and thus we have ${\bf{\Psi}}_{1} = {\bf{\Psi}}^{*}_{2}$.

According to Property 2, for the frequency pair ${\lambda _n} = \lambda _{M+n}^*$, the corresponding desired frequency response satisfies ${{\hat h}_n} = \hat h_{M+n}^*$ and thus, we also have $\hat {\bf{h}}_1 = \hat {\bf{h}}^{*}_2$. Meanwhile, for $2M+1 \le n \le N$, we have a real-valued ${\bf {\Psi}}_3$ and a real-valued $\hat {\bf{h}}_3$ since the corresponding frequencies $\lambda_n$ inside this range are real-valued.

Now, we can rewrite the solution of~\eqref{eq:FIRproblem} as
\[\begin{array}{l}
\hat{\bf g}  = {\bf{\Psi }}_{K + 1}^\dag \hat {\bf{h}}\\
=({\bf{\Psi}}_{1}^H {\bf{\Psi}}_{1} + {\bf{\Psi}}_{2}^H {\bf{\Psi}}_{2} +  {\bf{\Psi}}_{3}^T {\bf{\Psi}}_{3})^{-1}({\bf{\Psi}}_{1}^H \hat{\bf{h}}_1 \\
 + {\bf{\Psi}}_{2}^H \hat{\bf{h}}_2 + {\bf{\Psi}}_{3}^T \hat{\bf{h}}_3)\\
 =({\bf{\Psi}}_{1}^H {\bf{\Psi}}_{1} + {\bf{\Psi}}_{1}^T {\bf{\Psi}}_{1}^{*} +   {\bf{\Psi}}_{3}^T {\bf{\Psi}}_{3})^{-1}({\bf{\Psi}}_{1}^H \hat{\bf{h}}_1 \\
  + {\bf{\Psi}}_{1}^T \hat{\bf{h}}_1^{*} + {\bf{\Psi}}_{3}^T \hat{\bf{h}}_3 ).
\end{array}\]
It is obvious that ${{\bf{\Psi}}_{1}^H {\bf{\Psi}}_{1} + {\bf{\Psi}}_{1}^T {\bf{\Psi}}_{1}^{*}}$ and ${\bf{\Psi}}_{1}^H \hat{\bf{h}}_1 + {\bf{\Psi}}_{1}^T \hat{\bf{h}}_1^{*}$ are real-valued. Hence, solving~\eqref{eq:FIRproblem} leads to a real-valued solution.

\section{}
The error of the iterative approach on ${\boldsymbol{\alpha }}$ and $ {\boldsymbol{\beta}}$ is given by
\begin{equation}\label{eq:A1}
{\bf{e}}^{(i+1)} = {{{\boldsymbol{\gamma }}^{(i)}} \circ ({\hat {\bf h}} \circ {{\boldsymbol{\alpha }}}) - {{\boldsymbol{\beta }}} \circ {{\boldsymbol{\gamma }}^{(i)}}}
\end{equation}
By extending ${\boldsymbol{\alpha}}$ and $ {\boldsymbol{\beta}}$, we can rewrite \eqref{eq:A1} as
\begin{equation}\label{eq:A2}
{{\bf{e}}^{(i + 1)}} = {{\boldsymbol{\gamma }}^{(i)}} \circ {\hat {\bf h}} \circ ({{\bf{\Psi }}_{P + 1}}{{\bf{a}}}) - ({{\bf{\Psi }}_{Q + 1}}{{\bf{b}}}) \circ {{\boldsymbol{\gamma }}^{(i)}}
\end{equation}
The first term in the right hand side of \eqref{eq:A2} can be expressed as
\begin{equation}\label{eq:A3}
\begin{array}{l}
{{\boldsymbol{\gamma }}^{(i)}} \circ {\hat {\bf h}} \circ ({{\bf{\Psi }}_{P + 1}}{{\bf{a}}}) = {{\boldsymbol{\gamma }}^{(i)}} \circ [{\hat {\bf h}} \circ ({{\bf{\Psi }}_{P + 1}}{{\bf{a}}})]\\
 = {{\boldsymbol{\gamma }}^{(i)}} \circ \left\{ {[{{\bf{\Psi }}_{P + 1}} \circ ({\hat {\bf h}}{\bf 1}_{P + 1}^{\rm{T}})]{{\bf{a}}}} \right\}\\
 = [({{\boldsymbol{\gamma }}^{(i)}}{\bf{1}}_{P + 1}^{\rm{T}}) \circ {{\bf{\Psi }}_{P + 1}} \circ ({\hat {\bf h}}{\bf 1}_{P + 1}^{\rm{T}})]{{\bf{a}}}.
\end{array}
\end{equation}
Similarly, the second term in the right hand side of \eqref{eq:A2} is rewritten as
\begin{equation}\label{eq:A4}
\begin{split}
({{\bf{\Psi }}_{Q + 1}}{{\bf{b}}}) \circ {{\boldsymbol{\gamma }}^{(i)}}\! =\! [({{\boldsymbol{\gamma }}^{(i)}} {\bf 1}_{Q+1}^{\rm T})\circ {{\bf{\Psi }}_{Q + 1}}]{{\bf{b}}}.
\end{split}
\end{equation}
Finally, we define \eqref{eq:A3} and \eqref{eq:A4} as ${{\bf{H}}^{(i)}}{{\bf{a}}}$ and ${{\bf{B}}^{(i)}}{{\bf{b}}}$, respectively.
This trivially leads to \eqref{eq:newIT2}.

\bibliographystyle{IEEEtran}
\bibliography{main0520}

\begin{thebibliography}{10}
\providecommand{\url}[1]{#1}
\csname url@samestyle\endcsname
\providecommand{\newblock}{\relax}
\providecommand{\bibinfo}[2]{#2}
\providecommand{\BIBentrySTDinterwordspacing}{\spaceskip=0pt\relax}
\providecommand{\BIBentryALTinterwordstretchfactor}{4}
\providecommand{\BIBentryALTinterwordspacing}{\spaceskip=\fontdimen2\font plus
\BIBentryALTinterwordstretchfactor\fontdimen3\font minus
  \fontdimen4\font\relax}
\providecommand{\BIBforeignlanguage}[2]{{%
\expandafter\ifx\csname l@#1\endcsname\relax
\typeout{** WARNING: IEEEtran.bst: No hyphenation pattern has been}%
\typeout{** loaded for the language `#1'. Using the pattern for}%
\typeout{** the default language instead.}%
\else
\language=\csname l@#1\endcsname
\fi
#2}}
\providecommand{\BIBdecl}{\relax}
\BIBdecl

\bibitem{D}
D.~I. Shuman, S.~K. Narang, P.~Frossard, A.~Ortega, and P.~Vandergheynst, ``The
  emerging field of signal processing on graphs: Extending high-dimensional
  data analysis to networks and other irregular domains,'' \emph{IEEE Signal
  Processing Magazine}, vol.~30, no.~3, pp. 83--98, 2013.

\bibitem{AJM}
A.~Sandryhaila and J.~M. Moura, ``Big data analysis with signal processing on
  graphs: Representation and processing of massive data sets with irregular
  structure,'' \emph{IEEE Signal Processing Magazine}, vol.~31, no.~5, pp.
  80--90, 2014.

\bibitem{AJ}
------, ``Discrete signal processing on graphs,'' \emph{IEEE transactions on
  signal processing}, vol.~61, no.~7, pp. 1644--1656, 2013.

\bibitem{ASJ2}
------, ``Discrete signal processing on graphs: Frequency analysis.''
  \emph{IEEE Trans. Signal Processing}, vol.~62, no.~12, pp. 3042--3054, 2014.

\bibitem{chung1997spectral}
F.~R. Chung, \emph{Spectral graph theory}.\hskip 1em plus 0.5em minus
  0.4em\relax American Mathematical Soc., 1997, no.~92.

\bibitem{S}
S.~Chen, A.~Sandryhaila, J.~M. Moura, and J.~Kovacevic, ``Signal denoising on
  graphs via graph filtering,'' in \emph{Signal and Information Processing
  (GlobalSIP), 2014 IEEE Global Conference on}.\hskip 1em plus 0.5em minus
  0.4em\relax IEEE, 2014, pp. 872--876.

\bibitem{R3}
S.~Deutsch, A.~Ortega, and G.~Medioni, ``Manifold denoising based on spectral
  graph wavelets,'' in \emph{Acoustics, Speech and Signal Processing (ICASSP),
  2016 IEEE International Conference on}.\hskip 1em plus 0.5em minus
  0.4em\relax IEEE, 2016, pp. 4673--4677.

\bibitem{R4}
M.~Onuki, S.~Ono, M.~Yamagishi, and Y.~Tanaka, ``Graph signal denoising via
  trilateral filter on graph spectral domain,'' \emph{IEEE Transactions on
  Signal and Information Processing over Networks}, vol.~2, no.~2, pp.
  137--148, 2016.

\bibitem{zhang2008graph}
F.~Zhang and E.~R. Hancock, ``Graph spectral image smoothing using the heat
  kernel,'' \emph{Pattern Recognition}, vol.~41, no.~11, pp. 3328--3342, 2008.

\bibitem{AS}
A.~Sandryhaila and J.~M. Moura, ``Classification via regularization on
  graphs.'' in \emph{GlobalSIP}, 2013, pp. 495--498.

\bibitem{R6}
A.~Anis, A.~Gadde, and A.~Ortega, ``Efficient sampling set selection for
  bandlimited graph signals using graph spectral proxies,'' \emph{IEEE
  Transactions on Signal Processing}, vol.~64, no.~14, pp. 3775--3789, 2016.

\bibitem{SC1}
S.~Chen, A.~Sandryhaila, J.~M. Moura, and J.~Kova{\v{c}}evi{\'c}, ``Signal
  recovery on graphs: Variation minimization,'' \emph{IEEE Transactions on
  Signal Processing}, vol.~63, no.~17, pp. 4609--4624, 2015.

\bibitem{TPGV}
N.~Tremblay, G.~Puy, R.~Gribonval, and P.~Vandergheynst, ``Compressive spectral
  clustering,'' in \emph{Machine Learning, Proceedings of the Thirty-third
  International Conference (ICML 2016), June}, 2016, pp. 20--22.

\bibitem{AY}
S.~K. Narang and A.~Ortega, ``Perfect reconstruction two-channel wavelet filter
  banks for graph structured data,'' \emph{IEEE Transactions on Signal
  Processing}, vol.~60, no.~6, pp. 2786--2799, 2012.

\bibitem{R7}
D.~B. Tay and Z.~Lin, ``Design of near orthogonal graph filter banks,''
  \emph{IEEE Signal Processing Letters}, vol.~22, no.~6, pp. 701--704, 2015.

\bibitem{R1}
S.~K. Narang and A.~Ortega, ``Compact support biorthogonal wavelet filterbanks
  for arbitrary undirected graphs,'' \emph{IEEE transactions on signal
  processing}, vol.~61, no.~19, pp. 4673--4685, 2013.

\bibitem{DK}
D.~K. Hammond, P.~Vandergheynst, and R.~Gribonval, ``Wavelets on graphs via
  spectral graph theory,'' \emph{Applied and Computational Harmonic Analysis},
  vol.~30, no.~2, pp. 129--150, 2011.

\bibitem{R2}
A.~Sakiyama, K.~Watanabe, and Y.~Tanaka, ``Spectral graph wavelets and filter
  banks with low approximation error,'' \emph{IEEE Transactions on Signal and
  Information Processing over Networks}, vol.~2, no.~3, pp. 230--245, 2016.

\bibitem{R5}
D.~I. Shuman, C.~Wiesmeyr, N.~Holighaus, and P.~Vandergheynst,
  ``Spectrum-adapted tight graph wavelet and vertex-frequency frames,''
  \emph{IEEE Transactions on Signal Processing}, vol.~63, no.~16, pp.
  4223--4235, 2015.

\bibitem{shuman2011distributed}
D.~I. Shuman, P.~Vandergheynst, D.~Kressner, and P.~Frossard, ``Distributed
  signal processing via chebyshev polynomial approximation,'' \emph{IEEE
  Transactions on Signal and Information Processing over Networks}, 2018.

\bibitem{segarra2017optimal}
S.~Segarra, A.~Marques, and A.~Ribeiro, ``Optimal graph-filter design and
  applications to distributed linear network operators,'' \emph{IEEE
  Transactions on Signal Processing}, 2017.

\bibitem{Coutinio2017}
M.~Coutino, E.~Isufi, and G.~Leus, ``{Distributed Edge-Variant Graph
  Filters},'' in \emph{International Workshop on Computational Advances in
  Multi-Sensor Adaptive Processing (CAMSAP)}.\hskip 1em plus 0.5em minus
  0.4em\relax IEEE, 2017, pp. 1080--1084.

\bibitem{LWNT}
L.~Goldsberry, W.~Huang, N.~F. Wymbs, S.~T. Grafton, D.~S. Bassett, and
  A.~Ribeiro, ``Brain signal analytics from graph signal processing
  perspective,'' in \emph{IEEE International Conference on Acoustics, Speech,
  and Signal Processing, Submitted}, 2016.

\bibitem{EA}
E.~Isufi, A.~Loukas, A.~Simonetto, and G.~Leus, ``Autoregressive moving average
  graph filtering,'' \emph{IEEE Transactions on Signal Processing}, vol.~65,
  no.~2, pp. 274--288, 2017.

\bibitem{XM}
X.~Shi, H.~Feng, M.~Zhai, T.~Yang, and B.~Hu, ``Infinite impulse response graph
  filters in wireless sensor networks,'' \emph{IEEE Signal Processing Letters},
  vol.~22, no.~8, pp. 1113--1117, 2015.

\bibitem{PCVX}
D.~P. Bertsekas, \emph{Convex optimization theory}.\hskip 1em plus 0.5em minus
  0.4em\relax Athena Scientific Belmont, 2009.

\bibitem{SJA}
J.~R. Shewchuk \emph{et~al.}, ``An introduction to the conjugate gradient
  method without the agonizing pain,'' 1994.

\bibitem{M}
M.~H. Hayes, \emph{Statistical digital signal processing and modeling}.\hskip
  1em plus 0.5em minus 0.4em\relax John Wiley \& Sons, 2009.

\bibitem{edwards2004differential}
C.~H. Edwards, D.~E. Penney, and D.~T. Calvis, \emph{Differential equations and
  boundary value problems}.\hskip 1em plus 0.5em minus 0.4em\relax
  清华大学出版社, 2004.

\bibitem{sinswat1977eigenvalue}
V.~Sinswat and F.~Fallside, ``Eigenvalue/eigenvector assignment by
  state-feedback,'' \emph{International Journal of Control}, vol.~26, no.~3,
  pp. 389--403, 1977.

\bibitem{PA}
P.~Erdos and A.~R{\'e}nyi, ``On the evolution of random graphs,'' \emph{Publ.
  Math. Inst. Hung. Acad. Sci}, vol.~5, no.~1, pp. 17--60, 1960.

\bibitem{wilkinson1965algebraic}
J.~H. Wilkinson and J.~H. Wilkinson, \emph{The algebraic eigenvalue
  problem}.\hskip 1em plus 0.5em minus 0.4em\relax Clarendon Press Oxford,
  1965, vol.~87.

\bibitem{M2}
M.~E. Newman, ``The structure and function of complex networks,'' \emph{SIAM
  review}, vol.~45, no.~2, pp. 167--256, 2003.

\bibitem{DPP}
D.~I. Shuman, P.~Vandergheynst, and P.~Frossard, ``Chebyshev polynomial
  approximation for distributed signal processing,'' in \emph{Distributed
  Computing in Sensor Systems and Workshops (DCOSS), 2011 International
  Conference on}.\hskip 1em plus 0.5em minus 0.4em\relax IEEE, 2011, pp. 1--8.

\bibitem{CCH}
C.-I. Chang, ``Orthogonal subspace projection (osp) revisited: A comprehensive
  study and analysis,'' \emph{IEEE transactions on geoscience and remote
  sensing}, vol.~43, no.~3, pp. 502--518, 2005.

\bibitem{hu2017unified}
Y.~Hu and G.~Leus, ``On a unified framework for linear nuisance parameters,''
  \emph{EURASIP Journal on Advances in Signal Processing}, vol. 2017, no.~1,
  p.~4, 2017.

\bibitem{NYPM}
N.~Perraudin, J.~Paratte, D.~Shuman, L.~Martin, V.~Kalofolias,
  P.~Vandergheynst, and D.~K. Hammond, ``Gspbox: A toolbox for signal
  processing on graphs,'' \emph{arXiv preprint arXiv:1408.5781}, 2014.

\bibitem{SSG}
S.~P. Chepuri, S.~Liu, G.~Leus, and A.~O. Hero, ``Learning sparse graphs under
  smoothness prior,'' in \emph{Acoustics, Speech and Signal Processing
  (ICASSP), 2017 IEEE International Conference on}.\hskip 1em plus 0.5em minus
  0.4em\relax IEEE, 2017, pp. 6508--6512.

\bibitem{SAA}
S.~K. Narang, A.~Gadde, and A.~Ortega, ``Signal processing techniques for
  interpolation in graph structured data,'' in \emph{Acoustics, Speech and
  Signal Processing (ICASSP), 2013 IEEE International Conference on}.\hskip 1em
  plus 0.5em minus 0.4em\relax IEEE, 2013, pp. 5445--5449.

\bibitem{mao2014image}
Y.~Mao, G.~Cheung, and Y.~Ji, ``Image interpolation for dibr viewsynthesis
  using graph fourier transform,'' in \emph{3DTV-Conference: The True
  Vision-Capture, Transmission and Display of 3D Video (3DTV-CON), 2014}.\hskip
  1em plus 0.5em minus 0.4em\relax IEEE, 2014, pp. 1--4.

\bibitem{chepuri2017graph}
S.~P. Chepuri and G.~Leus, ``Graph sampling for covariance estimation,''
  \emph{IEEE Transactions on Signal and Information Processing over Networks},
  vol.~3, no.~3, pp. 451--466, 2017.

\bibitem{isufi2016separable}
E.~Isufi, A.~Loukas, A.~Simonetto, and G.~Leus, ``Separable autoregressive
  moving average graph-temporal filters,'' in \emph{Signal Processing
  Conference (EUSIPCO), 2016 24th European}.\hskip 1em plus 0.5em minus
  0.4em\relax IEEE, 2016, pp. 200--204.

\bibitem{isufi20162}
E.~Isufi, G.~Leus, and P.~Banelli, ``2-dimensional finite impulse response
  graph-temporal filters,'' in \emph{Signal and Information Processing
  (GlobalSIP), 2016 IEEE Global Conference on}.\hskip 1em plus 0.5em minus
  0.4em\relax IEEE, 2016, pp. 405--409.

\end{thebibliography}

\end{document}